\documentclass[final,leqno,onefignum,onetabnum]{siamltex1213}
\usepackage{mathtools}
\usepackage{amsmath,amssymb}
\usepackage{graphicx}
\usepackage{esint}

\title{Communicability Angle\\
and the Spatial Efficiency of Networks}

\author{Ernesto Estrada\footnotemark[2] \and Naomichi Hatano\footnotemark[3]}

\begin{document}
\maketitle

\global\long\def\thefootnote{\fnsymbol{footnote}}

\footnotetext[2]{Department of Mathematics \& Statistics, University of Strathclyde, Glasgow G11XQ, UK} 
\footnotetext[3]{Institute of Industrial Science, University of Tokyo, Komaba, Meguro, Tokyo 153-8505, Japan}

\global\long\def\thefootnote{\arabic{footnote}}

\begin{abstract}
We introduce the concept of communicability angle between a pair of
nodes in a graph. We provide strong analytical and empirical evidence
that the average communicability angle for a given network accounts
for its spatial efficiency on the basis of the communications among
the nodes in a network. We determine characteristics of the spatial
efficiency of more than a hundred real-world complex networks that
represent complex systems arising in a diverse set of scenarios. In
particular, we find that the communicability angle correlates very
well with the experimentally measured value of the relative packing efficiency
of proteins that are represented as residue networks. We finally show
how we can modulate the spatial efficiency of a network by tuning
the weights of the edges of the networks. This allows us to predict
effects of external stresses on the spatial efficiency of a network
as well as to design strategies to improve important parameters in
real-world complex systems. 
\end{abstract}
distance; graph planarity; Euclidean distance 

\noindent 

\pagestyle{myheadings} 
\thispagestyle{plain} 
\markboth{E.\ ESTRADA AND N.\ HATANO}{COMMUNICABILITY ANGLE}

\section{Introduction}

Graphs are frequently used to represent discrete objects both in abstract
mathematics and computer sciences as well as in applications, such
as theoretical physics, biology, ecology and social sciences~\cite{Handbook Graph Theory,Michaelbook}.
In the particular case of representing the networked skeleton of complex
systems, graphs receive the denomination of complex networks; we will
hereafter use graphs and networks interchangeably.

Complex networks are ubiquitous in many real-world
scenarios, ranging from the biomolecular --- those representing gene
transcription, protein interactions, and metabolic reactions --- to
the social and infrastructural organization of modern society~\cite{EstradaBook,NewmanReview,LucianoReview}.
In many of these networks, nodes and edges are used to represent physically
embedded objects~\cite{Spatial networks}, namely \textit{spatial
networks}. In urban street networks, for instance, the nodes describe
the intersection of streets, which are represented by the edges of
the graph. These streets and their intersections are embedded in the
two-dimensional space representing the surface occupied by the corresponding
city~\cite{Urban networks}. Thus, these networks
are planar graphs in the sense that we can draw them in a plane without
edge intersections, except for the few bridges and overpasses present
in a city. Another spatial network is the brain network,
in which the nodes account for brain regions embedded in the three-dimensional
space occupied by the brain, while the edges represent the communication
or physical connections among these regions~\cite{Brain networks}.
We can also capture the three-dimensional structure of proteins by
means of the residue networks in which nodes describe amino acids
and the edges represent physical interactions among them. Other examples
include the following: infrastructures, such as the Internet, transportation
networks, water and electricity supply networks, \textit{etc.}~\cite{Spatial networks};
anatomical networks, such as vascular and organ/tissue networks; the
networks of channels in fractured rocks; the networks representing
the corridors and galleries in animal nests; for even
more, see Ref.~\cite{EstradaBook} and references therein.

A natural question that arises in the analysis of spatial networks
is how efficiently they use the available geographical space in which
they are embedded. In a protein, for instance, the linear polypeptide
chain is folded up into a globular shape in order
to minimize the volume occupied inside the cell~\cite{Protein folding}.
In airport transportation networks the nodes are embedded into the
two-dimensional space represented by the surface of a country or continent,
but the connections between the airports occupy the available three-dimensional
space (it might be argued that they use a four-dimensional space as
two flights can intersect in space but at different times), which
increases the spatial efficiency of these networks. In contrast, the
planarity of urban street networks~\cite{Urban street planarity}
implies that both nodes and edges are embedded in a two-dimensional
space, which in general decreases the number of alternative routes
between different points in the network. This relatively poor spatial
efficiency of modern cities, \textit{i.e.}, the non-existence of three-dimensional
cities (although they have been already planned; see Chapter~3 in
Ref.~\cite{EstradaBook} and references therein), has posed a serious
challenge to their continuous growth in view of their threat to the
natural environment. Although the planarity may be an important part
of this problem, it is definitively not the only one. Two planar networks,
\textit{e.g.}, two cities, can display significantly different spatial
efficiency, and the same is true for pairs of non-planar networks.

The concept of \textit{spatial efficiency} is adapted here from economics,
where it is frequently used to describe how much time, effort and
cost a given arrangement produces for governments, businesses and
households to conduct their activities as compared to alternative
arrangements; see Ref.~\cite{Spatial efficiency} and references
therein. This concept has a lot to do with the efficiency in communication
among the parts of the system under study and as so it is a well-posed
problem for its analysis beyond spatial networks.

Indices for communication efficiency of networks have been previously
proposed in the literature~\cite{efficiency_1,efficiency_2,efficiency_3,efficiency_4}.
They have revolved around the idea of considering the sum of reciprocal shortest-path 
distances in graphs. It is worth mentioning that the sum of all
the reciprocal shortest-path distances in a graph is known in graph
theory as the Harary index, which was introduced by Plav\v{s}i\'{c} \textit{et
al}.\ in 1993~\cite{Harary_1} and studied elsewhere~\cite{Harary_2,Harary_3}.
The so-called efficiency index introduced by Latora and Marchiori~\cite{efficiency_1} 
(defined below in Eq.~\eqref{eq2.3-0}) is the average Harary index of a graph.
In the present work we will consider a new communication efficiency measure that
takes into account all the potential routes communicating a pair of
nodes instead of using the shortest paths only. The consequences of
this adoption will be developed in the rest of the paper.

In this context of communication among the nodes of a network, we~\cite{Communicability}
have introduced the \textit{communicability function} as a way to
quantify how much information can flow from one node to another in
a network; see also Refs.~\cite{Estrada Hatano Benzi,Estrada Higham}.
We regard the quantity $G_{pq}$, which we will define in Eq.~\eqref{eq2.20}
below, as the amount of information that departs from a node $p$
and ends at a node $q$. On the other hand, we regard $G_{pp}$
as the amount of information that departs from the original node $p$
and never arrives at the destination $q$, because it is returned
to its originator. Let us call the first amount of information the
successful information and the second the frustrated one. Then, the
goodness of communication between the two nodes is given by the
ratio of the successful to the frustrated amount of information.
Increasing the amount of successful information and reducing the amount
of frustrated one improves the quality of communication between the
two nodes. This has lead to the definition of a quantity~\cite{CommDist,CommDistAppl,Hyperspherical embedding}
that has been proved to be a distance between two nodes.

In the present paper, we show a remarkable mapping of each node of a
network to a point on the surface of a hypersphere. We prove that
the distance defined based on the communicability function is indeed
the chord distance between the two points on the
hypersphere. We can thereby assign a Euclidean angle to each pair
of nodes which represents the communication efficiency between them.
We then analyze various networks using the angle, which we refer to
as the communicability angle hereafter, and provide evidence that
this angle accounts for the spatial efficiency of networks.

\section{Preliminaries}

In this section we shall present some of the definitions, notations,
and properties associated with networks to make this work self-contained.
A \textit{graph} $\Gamma=(V,E)$ is defined by a set of $n$ nodes
(vertices) $V$ and a set of $m$ edges (links) $E=\{(p,q)|p,q\in V\}$
between the nodes. An edge is said to be \textit{incident} to a vertex
$p$ if there exists a node $q(\neq p)$ such that either $(p,q)\in E$
or $(q,p)\in E$. The \textit{degree} of a vertex, denoted by $k_{p}$,
is the number of edges incident to $p$ in $\Gamma$. The graph is
said to be \textit{undirected} if the edges are formed by unordered
pairs of vertices. A \textit{walk} of length $\ell$ in $\Gamma$
is a set of nodes $p_{1},p_{2},\ldots,p_{\ell},p_{\ell+1}$ such that
for all $1\leq i\leq\ell$, $(p_{i},p_{i+1})\in E$. A \textit{closed
walk} is a walk for which $p_{1}=p_{\ell+1}$. A \textit{path} is
a walk with no repeated nodes. A graph is \textit{connected} if there
is a path connecting every pair of nodes. A graph with unweighted
edges, no self-loops (edges from a node to itself), and no multiple
edges is said to be \textit{simple}. Throughout this work, we will
always consider undirected, simple, and connected networks.

More specifically, we will consider graphs which are defined as follows.
The path graph $P_{n}$ is a connected graph with
$n$ nodes, $n-2$ of which have degree 2 and the remaining two have
degree 1. The complete graph $K_{n}$ is the graph with $n$ nodes
and $n(n-1)/2$ edges. The complete bipartite graph $K_{n_{1},n_{2}}$
is the graph with $n=n_{1}+n_{2}$ nodes split into two disjoint sets,
one containing $n_{1}$ nodes and the other containing $n_{2}$ nodes,
while the edges connect every node in one set with every one in the
other. The particular case $K_{1,n-1}$ is known as the star graph.
A graph is planar if it can be drawn on a plane without any edge
crossings. The following is a well-known characterization of the planar
graphs known as the Kuratowski theorem (see Ref.~\cite{Topological graph theory}).

\begin{theorem} A network is planar if and only if it has no subgraph
homeomorphic to $K_{5}$ or $K_{3,3}$. \end{theorem}

Let us consider a matrix $A$ called the adjacency matrix, whose elements
are $A_{pq}=1$ if $(p,q)\in E$ and zero otherwise. For undirected
simple finite graphs, $A$ is a real symmetric matrix. 
We can therefore decompose it into the form 
\begin{align}
A=U\Lambda U^{T},
\end{align}
where $\Lambda$ is a diagonal matrix containing the eigenvalues of
$A$, which we label in non-increasing order $\lambda_{1}\geq\lambda_{2}\geq\ldots\geq\lambda_{n}$,
and $U=[\vec{\psi}_{1},\ldots,\vec{\psi}_{n}]$ is an orthogonal matrix,
where $\vec{\psi}_{\mu}$ is an eigenvector associated with $\lambda_{\mu}$.
Because we consider connected graphs, $A$ is irreducible; the Perron-Frobenius
theorem then dictates that $\lambda_{1}>\lambda_{2}$ and that we
can choose $\vec{\psi}_{1}$ 
such that its components $\psi_{1}(p)$ are positive for all $p\in V$.

An important quantity for studying communication processes in networks
is the communicability function~\cite{Communicability,Estrada Higham,Estrada Hatano Benzi},
defined for each pair of nodes $p$ and $q$ as 
\begin{align}
G_{pq}=\sum_{k=0}^{\infty}\frac{\left(A^{k}\right)_{pq}}{k!}=\left(e^{A}\right)_{pq}=\sum_{\mu=1}^{n}e^{\lambda_{\mu}}\vec{\psi}_{\mu}(p)\vec{\psi}_{\mu}(q).\label{eq2.20}
\end{align}
The factor $\left(A^{k}\right)_{pq}$ counts the number of walks of
length $k$ starting at the node $p$ and ending at the node $q$.
The communicability function is the sum of the numbers
of walks of length $k$, each weighted by the factor $1/k!$ so that
shorter walks may be more influential than longer ones.  In
Eq.~\eqref{eq2.20}, the exponential of the matrix $A$ is defined
by its Taylor expansion, which is the communicability function itself.
The spectral decomposition on the right-hand side is also derived
from the spectral decomposition of each term of the Taylor expansion:
$(A^{k})_{pq}=\sum_{\mu}(\lambda_{\mu})^{k}\vec{\psi}_{\mu}(p)\vec{\psi}_{\mu}(q)$. 

The importance of the communicability function~\eqref{eq2.20} lies
in the fact that it takes account of long walks as
well as short ones; even two nodes connected by a very long shortest
path can have a strong communication if they are connected by very
many longer walks. The diagonal term $G_{pp}$ characterizes the degree
of participation of the node $p$ in all subgraphs of the network.
It is thus known as the subgraph centrality of the corresponding node~\cite{Subgraph Centrality}.

We can visualize the communicability function~\eqref{eq2.20}
in another way. Consider a matrix-vector equation $d\vec{\psi}/dt=A\vec{\psi}$,
which governs the time evolution of a vector $\vec{\psi}(t)$. If
the vector $\vec{\psi}(t)$ describes a random-walker distribution
on the network in question at time $t$, the above equation describes
how the walkers move around on the network. Its formal solution is
given by $\vec{\psi}(t)=e^{At}\vec{\psi}(0)$, and hence the exponential
matrix $e^{A}$ is the time evolution operator for the unit time.
Therefore, the communicability function~\eqref{eq2.20} is the transition rate
for the walkers on the site $p$ (represented by a vector $\vec{w}_{p}$)
to move to the site $q$ (represented by another vector $\vec{w}_{q}$)
after the unit time, where $w_{p}$ is a column vector with unity
on the $p$th element and zero on the others. 

It is possible to define several distance measures on networks. The
most common one is the \textit{shortest-path} or \textit{geodesic
distance} between two nodes $p,q\in V$, which is defined as the length
of the shortest path connecting these nodes. We will write $d(p,q)$
to denote the distance between $p$ and $q$. Here we will refer to
the average of the shortest-path distance in the graph as the average
path length, as usual in network theory. 
The communication efficiency of a networks is defined on the basis
of this shortest-path distance as~\cite{efficiency_1}
\begin{equation}\label{eq2.3-0}
E=\dfrac{1}{n\left(n-1\right)}\sum_{p,q}d_{pq}^{-1},
\end{equation}
where $d_{pq}=d\left(p,q\right)$. The index $H=\sum_{p,q}d_{pq}^{-1}$
is the Harary index~\cite{Harary_1} mentioned above.

Another distance measure
among the nodes of a graph is the so-called resistance distance~\cite{Resistance distance}
which is defined by 
$\Omega_{pq}=L_{pp}^{+}+L_{qq}^{+}-2L_{pq}^{+}$,
where $L^{+}$ is the Moore-Penrose pseudoinverse of
the Laplacian matrix of the network~\cite{resistance and Laplacian,Minimising effective resistance};
the network Laplacian is defined by $L=K-A$ with $K=\mathop{\mathrm{diag}}\left(k_{i}\right)$.
The resistance distance considers not only the shortest paths but
also longer walks in the communication between two nodes. In spite of 
the potential similarity with the communicability function~\eqref{eq2.20}, they exhibit
very important differences. For instance, the communicability distance,
a metric based on communicability function, is highly uncorrelated
with the resistance distance~\cite{CommDist}. More importantly,
Luxburg \textit{et al}.~\cite{lost_in_space} have proved that for extremely
large graphs the resistance distance converges to an expression that
does not take into account the structure of the graph at all and 
is completely meaningless as a distance function on the graph. This
situation does not occur for the communicability-based functions.

An important concept in graph theory is the isoperimetric
number of the graph, which is the discrete analogous of the Cheeger
constant. Let $S\subset V$, such that $0<\left|S\right|<\left|V\right|/2$.
Also, let $\partial S=\left\{ U\subset E\left|\left(p,q\right)\in U\implies p\in S,q\in\bar{S}\right.\right\} $
be the neighborhood of the set $S$. The isoperimetric constant is
defined as
\begin{equation}\label{eq2.3}
i\left(G\right)=\underset{S}{\inf}\dfrac{\left|\partial S\right|}{\left|S\right|}.
\end{equation}
A large isoperimetric number indicates that the graph
lacks any structural bottlenecks, which means that the graph is super-homogeneous,
lacking any holes or core-periphery structures. Formally, a graph hole, also
known as a chordless cycle, is a cycle $C$ of length at least four
such that no two nodes of the cycle are connected by an edge that
does not itself belong to $C$. 

In the next section we will introduce a third distance measure
defined recently on the basis of the communicability function. It
is novel in the sense that longer walks than the shortest path are
taken into account.

\section{Communicability distance}

The new distance function is defined as~\cite{CommDist,CommDistAppl}
\begin{align}
{\xi_{pq}}^{2}=G_{pp}+G_{qq}-2G_{pq},\label{eq3.10}
\end{align}
which we will refer to as the communicability distance between the
nodes $p$ and $q$ in $\Gamma$. The intuition behind it is that
when two nodes $p$ and $q$ communicate with each other, the quality
of their communication depends on two factors: (i) how much information
departing from the node $p$ ($q$) arrives at the node $q$ ($p$),
and (ii) how much information departing from the node $p$ ($q$)
returns to that node $p$ ($q$) without arriving at its destination.
That is, the communication efficiency increases with the amount of
information which departs from the originator and arrives at its destination,
but decreases with the amount of information which is frustrated due
to the fact that the information returns to its originator without
being delivered to its target. We can rephrase the
information flow as the random walkers according to the interpretation
that $e^{A}$ is a time-evolution operator for the unit time. This
intuition has lead to the definition~\eqref{eq3.10}.

It has been indeed proved that the function $\xi_{pq}$ is a Euclidean
distance between the nodes $p$ and $q$ in $\Gamma$~\cite{CommDist}.

\renewcommand{\thetheorem}{\arabic{section}.\arabic{theorem}~\cite{Hyperspherical embedding}}
 \begin{theorem} The communicability distance $\xi_{pq}$ induces
an embedding of the graph $\Gamma$ of size $n$ into a hypersphere
of radius $R^{2}=[c-{\left(2-b\right)^{2}}/{a}]/4$ in an $\left(n-1\right)$-dimensional
space, where $a=\vec{1}^{T}e^{-A}\vec{1}$, $b=\vec{s}^{T}e^{-A}\vec{1}$
and $c=\vec{s}^{T}e^{-A}\vec{s}$ with $\vec{s}=\mathop{\mathrm{diag}}e^{A}$.
\end{theorem} 
\renewcommand{\thetheorem}{\arabic{section}.\arabic{theorem}}

Let us hereafter give a more intuitive and geometric
view of the communicability distance. For the purpose, we first prove
the following theorem.

\begin{theorem} Let $\vec{x}_{p}=e^{\Lambda/2}\vec{\phi}_{p}$, where
$\vec{\phi}_{p}=\begin{pmatrix}\psi_{1}(p) & \cdots & \psi_{\mu}(p) & \cdots & \psi_{n}(p)\end{pmatrix}^{T}$.
Then we have 
\begin{align}
G_{pq}=\vec{x}_{p}\cdot\vec{x}_{q}.\label{eq3.20}
\end{align}
\end{theorem} \begin{proof} Let 
$X=\begin{pmatrix}\vec{x}_{1} & \cdots & \vec{x}_{p} & \cdots & \vec{x}_{n}\end{pmatrix}=e^{\Lambda/2}U^{T}$.
We therefore have 
\begin{align}
X^{T}X=Ue^{\Lambda}U^{T}=e^{A}=G,
\end{align}
which is immediately followed by Eq.~\eqref{eq3.20}. \end{proof}

This theorem transforms the communicability distance~\eqref{eq3.10}
into the form 
\begin{align}
{\xi_{pq}}^{2}=\vec{x}_{p}\cdot\vec{x}_{p}+\vec{x}_{q}\cdot\vec{x}_{q}-2\vec{x}_{p}\cdot\vec{x}_{q}=\left(\vec{x}_{p}-\vec{x}_{q}\right)^{2}.
\end{align}
In other words, the communicability distance is the Euclidean distance
in the space of $\{\vec{x}_{p}\}$. In order to visualize
this space,  let us go back to the interpretation
that $(e^{A})_{pq}$ is the transition rate of the random walkers
from the $p$th site to the $q$th site. An initial state $\vec{w}_{p}$
is a basis vector in the original vector space but its expression
is given by $\vec{\phi}_{p}$ above in the vector space with the eigenvectors
$\vec{\psi}_{\mu}$ as its basis vectors, namely the eigenspace; see
Fig.~\ref{fig1} for an example for the path graph $P_3$. 
\begin{figure}
\centering \includegraphics[width=0.4\textwidth]{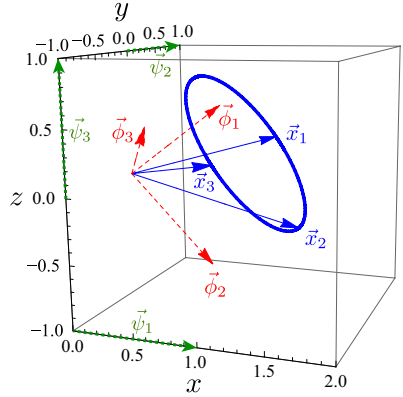} 
\caption{A demonstration plot of $\vec{\phi}_{p}$ (red dashed arrows) and
$\vec{x}_{p}$ (blue solid arrows) for the path graph $P_{3}$. The
eigenvectors $\vec{\psi}_{\mu}$ (green dotted arrows) define the
axes of this eigenspace. The communicability distance $\xi_{pq}$
is the chord distance on the (blue) circle that goes through the end
points of the vectors $\vec{x}_{p}$.}
\label{fig1} 
\end{figure}
In other words, the initial vector represents the state in which all
random walkers sit on the $p$th site, but it is denoted by the vector
$\vec{\phi}_{p}$ in the eigenspace. The vector $\vec{x}_{p}$ is
a vector in the eigenspace, representing a state in which random walkers
from the $p$th site move around for the time $1/2$. 

Theorem 3.1 dictates that the vectors $\{\vec{x}_{p}\}$ fall onto
the surface of a hypersphere in the space; see Fig.~\ref{fig2}(a)
for illustration in the case $n=3$. 
\begin{figure}
\centering %
\includegraphics[width=0.45\textwidth]{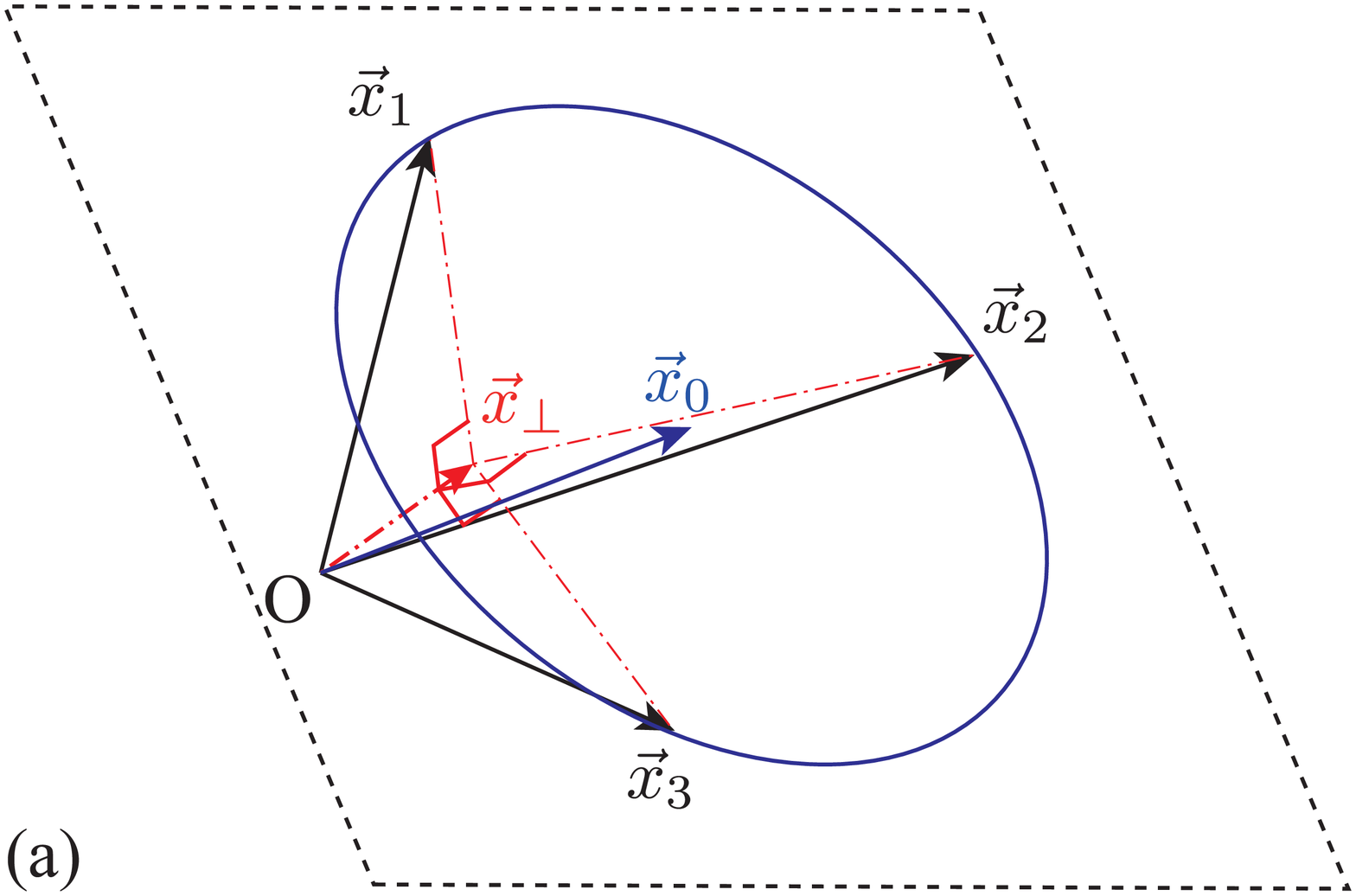}%
\hspace{0.1\textwidth}
\includegraphics[width=0.3\textwidth]{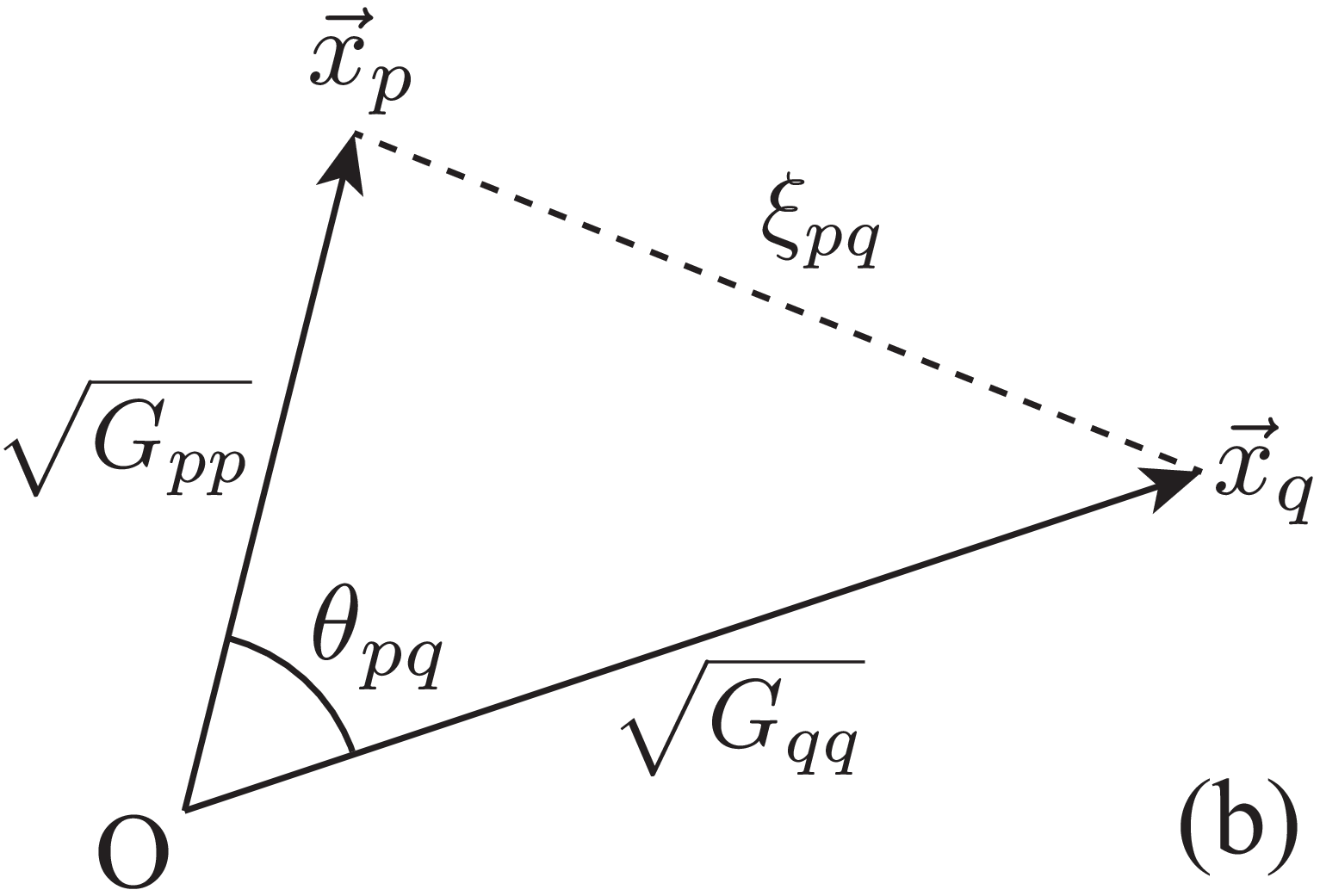} %
\caption{(a) Three vectors $\vec{x}_{1}$, $\vec{x}_{2}$ and $\vec{x}_{3}$
(solid black arrows) in a three-dimensional space spanned by the three
eigenvectors of a $3\times3$ adjacency matrix $A$. The vectors fall
on a two-dimensional flat surface (broken black lines) to which the
vector $\vec{x}_{\perp}$ (red dot-dashed arrow) is normal. We can
draw a circle (solid blue curve) on the two-dimensional surface around
a point $\vec{x}_{0}$ (solid blue arrow) to contain all three points.
(b) The triangle spanned by the vectors $\vec{x}_{p}$ and $\vec{x}_{q}$.}
\label{fig2} 
\end{figure}
We can understand this in the following way. We first fix the $n$-dimensional
normal vector $\vec{x}_{\perp}$ from $n$ pieces of conditions $(\vec{x}_{p}-\vec{x}_{\perp})\cdot\vec{x}_{\perp}=0$
for $1\leq p\leq n$. It specifies the $(n-1)$-dimensional flat surface
on which all vectors fall as $(\vec{x}-\vec{x}_{\perp})\cdot\vec{x}_{\perp}=0$.
We next fix the $n$-dimensional vector $\vec{x}_{0}$ that specifies
the center of the hypersphere as well as the radius $R$ from $n+1$
pieces of conditions $(\vec{x}_{0}-\vec{x}_{\perp})\cdot\vec{x}_{\perp}=0$
and $|\vec{x}_{p}-\vec{x}_{0}|=R$ for $1\leq p\leq n$.

We can therefore regard ${\xi_{pq}}$ as the chord
distance between the two points on the hypersurface. Figure~\ref{fig2}(b)
picks out the triangle spanned by the vectors $\vec{x}_{p}$ and $\vec{x}_{q}$.
This leads to the definition in the next section of the angle between
the two vectors.

\section{Communicability angle}

Let $p$ and $q$ be nodes of a connected simple network and let us
define the following quantity: 
\begin{align}
\gamma_{pq} & \coloneqq\frac{G_{pq}}{\sqrt{G_{pp}G_{qq}}}.
\end{align}
We then prove the following result. \begin{theorem} The index $\gamma_{pq}$
is the cosine of the Euclidean angle spanned by the position vectors
of $p$ and $q$. \end{theorem} \begin{proof} The view shown in
Fig.~\ref{fig2}(b) obviously gives 
\begin{align}
\cos\theta_{pq} & =\frac{\vec{x}_{p}\cdot\vec{x}_{q}}{\left|\vec{x}_{p}\right|\left|\vec{x}_{q}\right|}.
\end{align}
The use of Eq.~\eqref{eq3.20} then proves the result. \end{proof}

We then call $\theta_{pq}$ the communicability angle between the
corresponding nodes of the graph. Details on how to compute the
communicability angle for networks are given in the Supplementary Information
accompanying the present paper.
For each pair of nodes in the graph,
the communicability distance and angle are related mathematically
by the following expression: 
\begin{align}\label{eq4.3}
{\xi_{pq}}^{2}=G_{pp}+G_{qq}-2\sqrt{G_{pp}G_{qq}}\cos\theta_{pq}.
\end{align}
Because $G_{pq}\geq0$ for any pair of nodes in $\Gamma$, the communicability
angle is bounded by $0\leq\cos\theta_{pq}\leq1$. That is, the communicability
angle of simple graphs can take values only in the range $(0^{\circ},90^{\circ})$.
We will now give classes of graphs that show how we attain the extremal
values.

\begin{proposition} \label{prop4.1} Let $P_{n}$ be the path graph
with $n$ nodes labeled by $1,2,\cdots,n$ sequentially. The communicability
angle between any pair of nodes in $P_{n}$ is given by 
\begin{align}
\cos\theta_{pq}\left(P_{n}\right)=\frac{I_{p-q}(2)-I_{p+q}(2)}{\sqrt{\left[I_{0}(2)-I_{2r(p)}(2)\right]\left[I_{0}(2)-I_{2r(q)}(2)\right]}}\label{eq4.40}
\end{align}
in the limit $n\to\infty$, where $I_{\gamma}\left(z\right)$ is the
Bessel function of the first kind and 
\begin{align}
r(p)=\begin{cases}
p & \mbox{for \ensuremath{p\leq n/2} with even \ensuremath{n} or \ensuremath{p\leq\left(n+1\right)/2} with odd \ensuremath{n},}\\
n-p+1 & \mbox{for \ensuremath{p>n/2} with even \ensuremath{n} or \ensuremath{p>\left(n+1\right)/2} with odd \ensuremath{n}.}
\end{cases}
\end{align}
\end{proposition} \begin{proof} The eigenvalues and eigenvectors
of the adjacency matrix of $P_{n}$ are 
\begin{align}
\lambda_{j}\left(P_{n}\right)=2\cos\frac{j\pi}{n+1},\qquad\psi_{j}\left(p\right)=\sqrt{\frac{2}{n+1}}\sin\frac{jp\pi}{n+1}
\end{align}
for $1\leq j\leq n$. Thus 
\begin{align}
G_{pq}\left(P_{n}\right) & =\frac{1}{n+1}\sum_{j=1}^{n}\left[\cos\frac{j\pi(p-q)}{n+1}-\cos\frac{j\pi(p+q)}{n+1}\right]e^{2\cos\left(j\pi/(n+1)\right]},\\
G_{pp}\left(P_{n}\right) & =\frac{1}{n+1}\sum_{j=1}^{n}\left[1-\cos\frac{2j\pi p}{n+1}\right]e^{2\cos\left(j\pi/(n+1)\right]}.
\end{align}
In the limit $n\rightarrow\infty$, we can write them in integral
forms, which eventually reduce to 
$G_{pq}\left(P_{n}\right)=I_{p-q}(2)-I_{p+q}(2)$,
and 
$G_{pp}\left(P_{n}\right)=I_{0}(2)-I_{2r(p)}(2)$; 
this proves Eq.~\eqref{eq4.40}. \end{proof}

Notice that for the pair of nodes at the ends of the path we have
\begin{align}
\lim_{n\to\infty}\cos\theta_{n1}\left(P_{n}\right)=\lim_{n\to\infty}\frac{I_{n-1}(2)-I_{n+1}(2)}{I_{0}(2)-I_{2}(2)}=0,
\end{align}
which attains the lower bound of the communicability angle.

\begin{proposition} \label{prop4.2} Let $K_{1,n-1}$ be the star
graph with $n$ nodes. Let the node with degree $n-1$ labelled as
1. The communicability angle between any pair of nodes in $K_{1,n-1}$
is given by 
\begin{align}
\cos\theta_{1q}\left(K_{1,n-1}\right) & =\frac{\tanh^{2}\left(\sqrt{n-1}\right)}{\left(n-2\right)\textnormal{sech\ensuremath{\left(\sqrt{n-1}\right)}}+1}\quad\mbox{for \ensuremath{q\neq1}},\label{eq4.120}\\
\cos\theta_{pq}\left(K_{1,n-1}\right) & =\frac{\cosh\left(\sqrt{n-1}\right)-1}{\left(n-2\right)\cosh\left(\sqrt{n-1}\right)+n-2}\quad\mbox{for \ensuremath{p\neq1} and \ensuremath{q\neq1}}.\label{eq4.130}
\end{align}
\end{proposition} \begin{proof} The communicability between the
different pairs of nodes in $K_{1,n-1}$ are 
\begin{align}
G_{1q}\left(K_{1,n-1}\right) & =\frac{1}{\sqrt{n-1}}\sinh\left(\sqrt{n-1}\right)\quad\mbox{for \ensuremath{q\neq1}},\\
G_{pq}\left(K_{1,n-1}\right) & =\frac{1}{n-1}\left[\cosh\left(\sqrt{n-1}\right)-1\right]\quad\mbox{for \ensuremath{p\neq1} and \ensuremath{q\neq1}}.
\end{align}
The subgraph centrality of the two distinct nodes in the star graph
are 
\begin{align}
G_{11}\left(K_{1,n-1}\right) & =\cosh\left(\sqrt{n-1}\right),\\
G_{pp}\left(K_{1,n-1}\right) & =\frac{1}{n-1}\left[\cosh\left(\sqrt{n-1}\right)+n-2\right]\quad\mbox{for \ensuremath{p\neq1}}.
\end{align}
Algebra with trigonometric identities gives Eqs.~\eqref{eq4.120}
and~\eqref{eq4.130}. \end{proof}

It is important to notice that 
\begin{align}
\lim_{n\to\infty}\cos\theta_{1q}\left(K_{1,n-1}\right) & =1\quad\mbox{for \ensuremath{q\neq1}},\\
\lim_{n\to\infty}\cos\theta_{pq}\left(K_{1,n-1}\right) & =1\quad\mbox{for \ensuremath{p\neq1} and \ensuremath{q\neq1}},
\end{align}
which attain the upper bound of the communicability angle.

\begin{proposition} \label{prop4.3} Let $K_{n}$ be the complete
graph with $n$ nodes. The communicability angle between any pair
of nodes in $K_{n}$ is given by 
\begin{align}
\cos\theta_{pq}=\frac{e^{n}-1}{e^{n}+n-1}.\label{eq:Kn}
\end{align}
\end{proposition} \begin{proof} The eigenvalues of the adjacency
matrix of $K_{n}$ are $n-1$ with multiplicity $1$ and $-1$ with
multiplicity $n-1$. We thereby have 
\begin{align}
G_{pp}=\frac{1}{ne}\left(e^{n}+n-1\right),\qquad G_{pq}=\frac{1}{ne}\left(e^{n}-1\right)
\end{align}
which proves Eq.~\eqref{eq:Kn}. \end{proof}

Notice that $\cos\theta_{pq}\rightarrow1$ as $n\rightarrow\infty$
in $K_{n}$.

\section{Communicability distance and communicability angle}

An interesting difference between the communicability distance $\xi_{pq}$
and the communicability angle $\theta_{pq}$ arises from their analysis
in a path $P_{n}$. First, we prove the following result for the communicability
distance.

\begin{proposition} \label{prop4.4} Let $P_{n}$ be a path graph
of $n$ nodes labeled consecutively from one end point to the other
as $1,2,\cdots,n$. Let $S=\left\{ {\xi_{12}}^{2},{\xi_{13}}^{2},\cdots,{\xi_{1n}}^{2}\right\} $
be the ordered sequence of communicability distances between the first
node and any other nodes $q$ in the path. Then, $S$ is nonmonotonic.
\end{proposition} \begin{proof} Without any loss of generality we
will consider here even $n$ for simplicity. The communicability distance
in question is given by 
\begin{align}
{\xi_{1q}}^{2}=\begin{cases}
\left[2I_{0}(2)-I_{2}(2)\right]-\left[I_{2q}(2)+2I_{1-q}(2)-2I_{1+q}(2)\right] & \mbox{for \ensuremath{1<q\leq n/2}},\\
\left[2I_{0}(2)-I_{2}(2)\right]-\left[I_{2(n-q+1)}(2)+2I_{1-q}(2)-2I_{1+q}(2)\right] & \mbox{for \ensuremath{q>n/2}},
\end{cases}
\end{align}
where $r(p)$ and $I_{\gamma}(z)$ are as before. First, we have 
\begin{align}
{\xi_{12}}^{2}\simeq1.0637
\end{align}
in the limit $n\to\infty$. Next, let $\chi(q)=I_{2q}(2)+2I_{1-q}(2)-2I_{1+q}(2)$.
It is easy to check that $\chi(q)>\chi(q+1)$, so that ${\xi_{1q}}^{2}$
increases as $q\to n/2$. For nodes relatively close to the center
of the path, we have 
\begin{align}
\lim_{q\to n/2}{\xi_{1q}}^{2}=2I_{0}(2)-I_{2}(2)\approx3.8702,
\end{align}
but as $q$ approaches the other end of the path, we have 
\begin{align}
\lim_{q\to n}{\xi_{1q}}^{2}=2I_{0}(2)-2I_{2}(2)\approx3.1813.
\end{align}
This means that the communicability distances increases from $\xi_{12}$
up to the maximum $\xi_{1q}\approx3.8702$ and then decreases to $\xi_{1n}\approx3.1813$,
which proves the result. \end{proof}

We now prove that the monotonicity holds for the communicability angle.

\begin{proposition} \label{prop4.5} Let $P_{n}$ be a path graph
of $n$ nodes labeled consecutively from one end point to the other
as $1,2,\cdots,n$. Let $C=\left\{ \theta_{12},\theta_{13},\cdots,\theta_{1n}\right\} $
be the ordered sequence of communicability angles between the first
node and any other nodes $q$ in the path. Then, $C$ is monotonic.
\end{proposition} \begin{proof} Without any loss of generality we
will consider here again even $n$. The communicability angle in question
is given by: 
\begin{align}
\cos\theta_{1q}=\begin{cases}
\dfrac{I_{1-q}(2)-I_{1+q}(2)}{\sqrt{\left[I_{0}(2)+I_{2}(2)\right]\left[I_{0}(2)-I_{2q}(2)\right]}} & \mbox{for \ensuremath{1<q\leq n/2}},\\
\dfrac{I_{1-q}(2)-I_{1+q}(2)}{\sqrt{\left[I_{0}(2)+I_{2}(2)\right]\left[I_{0}(2)-I_{2(n-q+1)}(2)\right]}} & \mbox{for \ensuremath{q>n/2}}.
\end{cases}\label{eq:Pn}
\end{align}
For small values of $q$ it is easy to see that $\cos\theta_{1q}>\cos\theta_{1,q+1}$;
the numerator of \eqref{eq:Pn} decreases as $q$ increases and at
the same time the denominator decreases. It is also easy to see that
$\lim_{q\to\infty}\cos\theta_{1q}=0$.

The difference with the result for the communicability distance arises
from the fact that the numerator of \eqref{eq:Pn} for $q>n/2$ is
the same as that for $1<q\leq n/2$. We therefore have $\lim_{q\to\infty}\cos\theta_{1q}=0$
for $q>n/2$, which indicates that once the angle between the first
and the $q$th nodes in $P_{n}$ reaches its maximum value, \textit{i.e.},
$90^{\circ}$, it does not decrease again, which proves that the series
$C$ is monotonic. \end{proof}

Now, let us extract the structural information provided by these results
which will be useful for further application of the communicability
angle in analyzing real-world complex networks. Let us define the
average communicability angle for a given graph as the average over
the pairs of nodes: 
\begin{align}
\left\langle \theta\right\rangle =\dfrac{2}{n\left(n-1\right)}\sum_{p>q}\theta_{pq}.\label{eq5.60}
\end{align}
We then have the following observations: 
(i) The average communicability angle for the path graph $P_{n}$ tends
to $90^{\circ}$ when the number of nodes tends to infinite. This
is a consequence of Propositions~\ref{prop4.1} and~\ref{prop4.5}.
(ii) The average communicability angle for the star graph $K_{1,n-1}$
tends to $0^{\circ}$ when the number of nodes tends to infinite.
This is a consequence of Proposition~\ref{prop4.2}.
(iii) The average communicability angle for the complete graph $K_{n}$
tends to $0^{\circ}$ when the number of nodes tends to infinite.
This is a consequence of Proposition~\ref{prop4.3}.

\section{Computational analysis of the communicability angle}

\label{sec6}

In this section we computationally analyze the average communicability
angle $\left\langle\theta\right\rangle$ in Eq.~\eqref{eq5.60} for connected graphs. Specifically, we here
study a dataset of all 11,117 connected graphs with 8 nodes. We divide
this section into three subsections: we first analyze relations (or
lack thereof) between the average communicability angle and other
graphs metrics, namely the average path length, the average resistance
distance and the average communicability distance; we then study relations
between $\left\langle \theta\right\rangle $ and the graph planarity;
we finally investigate influence of graph modularity on the communicability
angle.

\subsection{Communicability angle and other graph metrics}

We first compare the average communicability angle $\left\langle \theta\right\rangle $
with the average communicability distance $\left\langle \xi\right\rangle $,
the average resistance distance $\left\langle \Omega\right\rangle $, 
the average path length $\left\langle l\right\rangle $ 
and the communication efficiency $E$ as metrics potentially related to $\left\langle \theta\right\rangle $;
every average was taken over all pairs of nodes. We show in Fig.~\ref{fig3}
the scatter plots of these measures against the average communicability
angles.
\begin{figure}
 \includegraphics[width=0.31\textwidth]{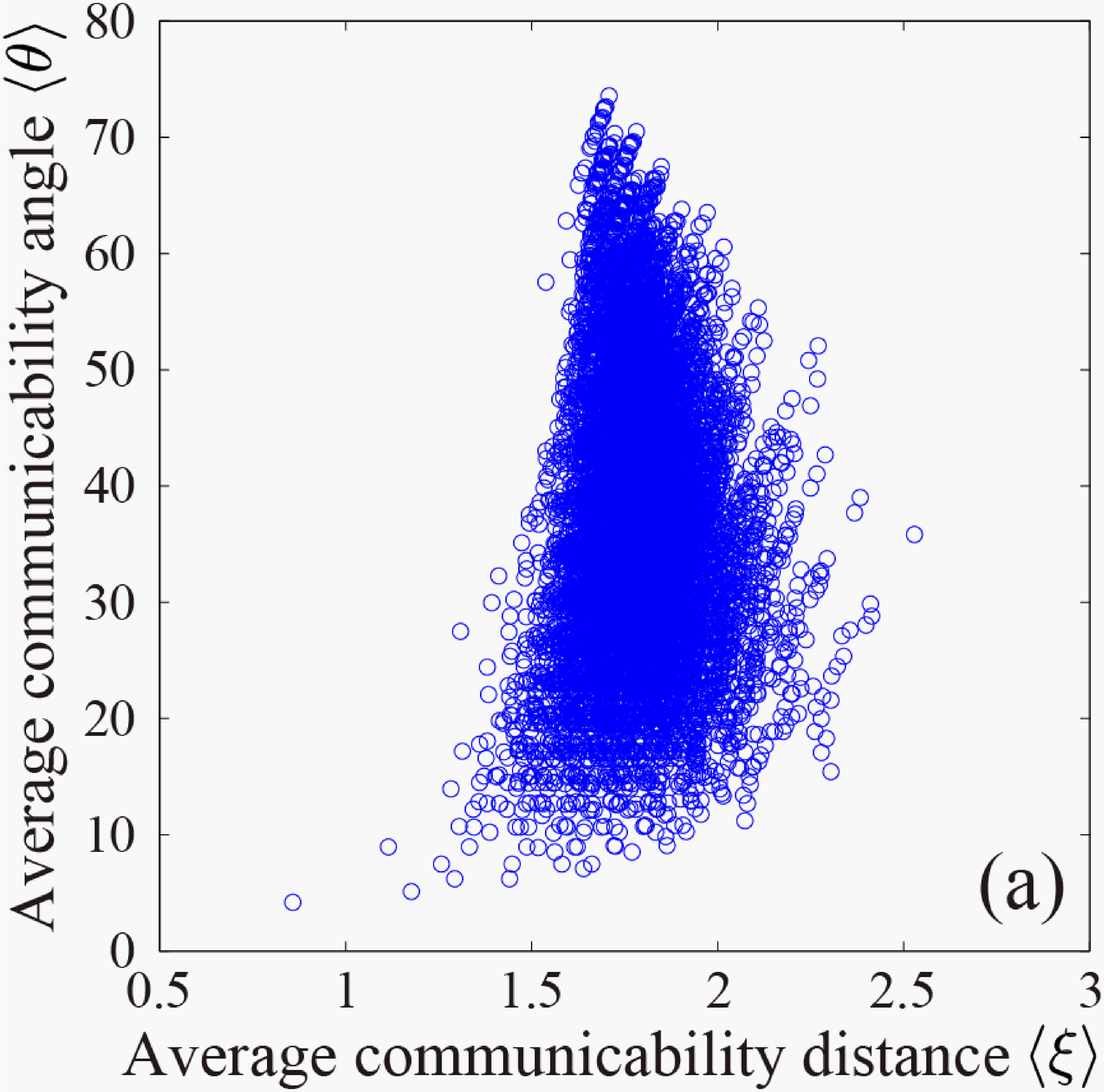} %
\hspace{0.03\textwidth} %
 \includegraphics[width=0.31\textwidth]{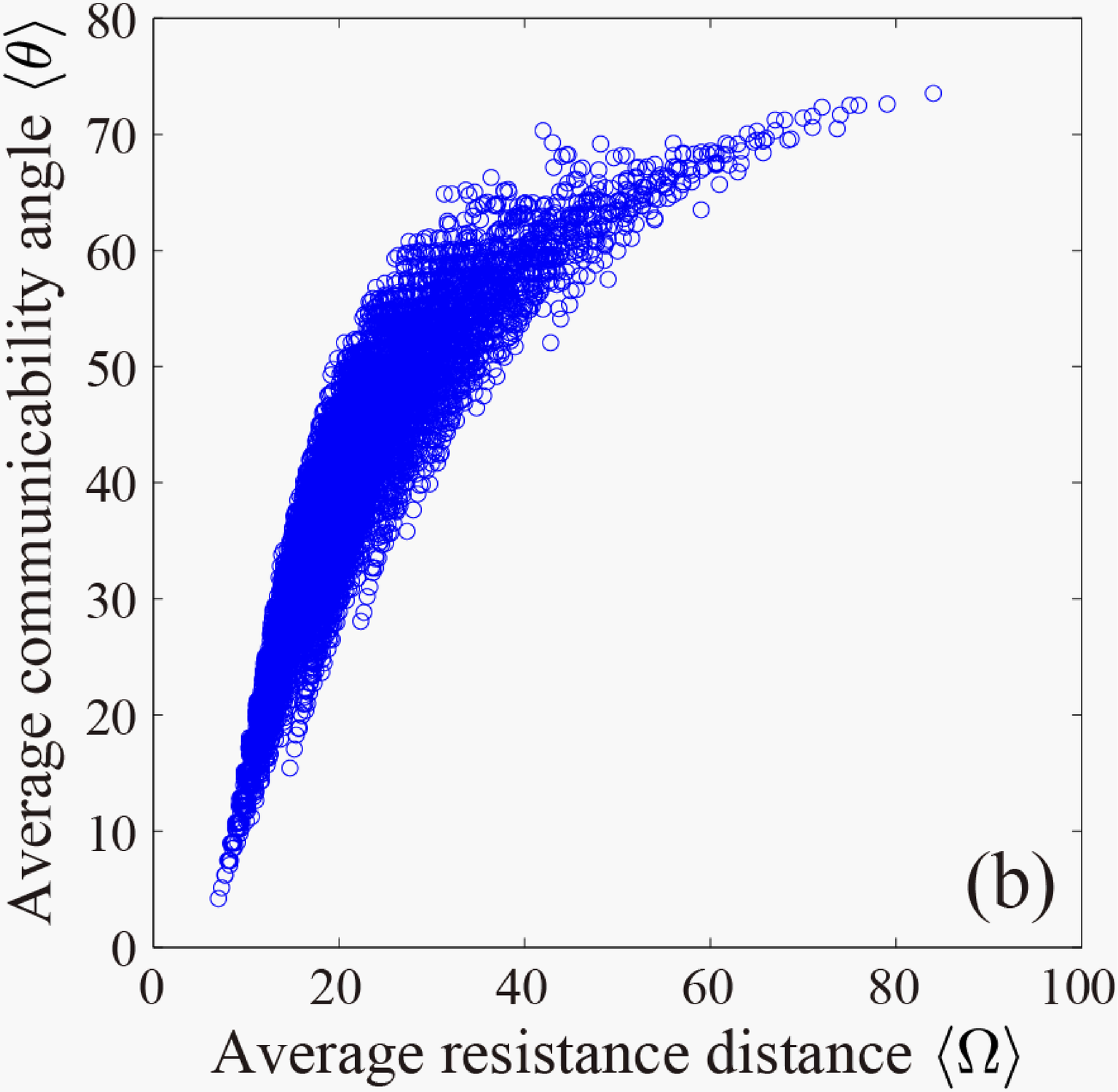} %
\hfill %
 \includegraphics[width=0.31\textwidth]{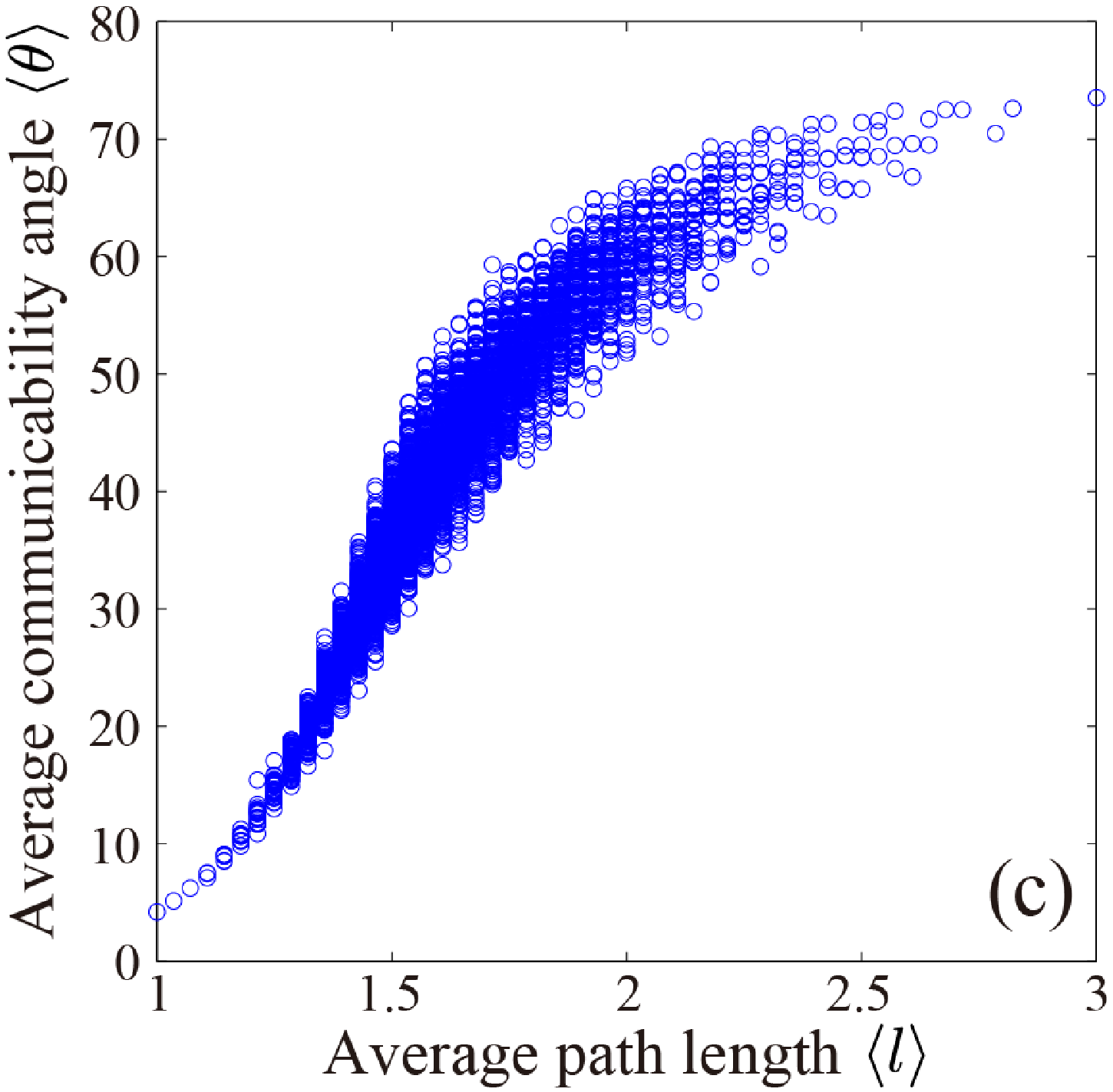}
\vspace{\baselineskip}
\\
 \includegraphics[width=0.31\textwidth]{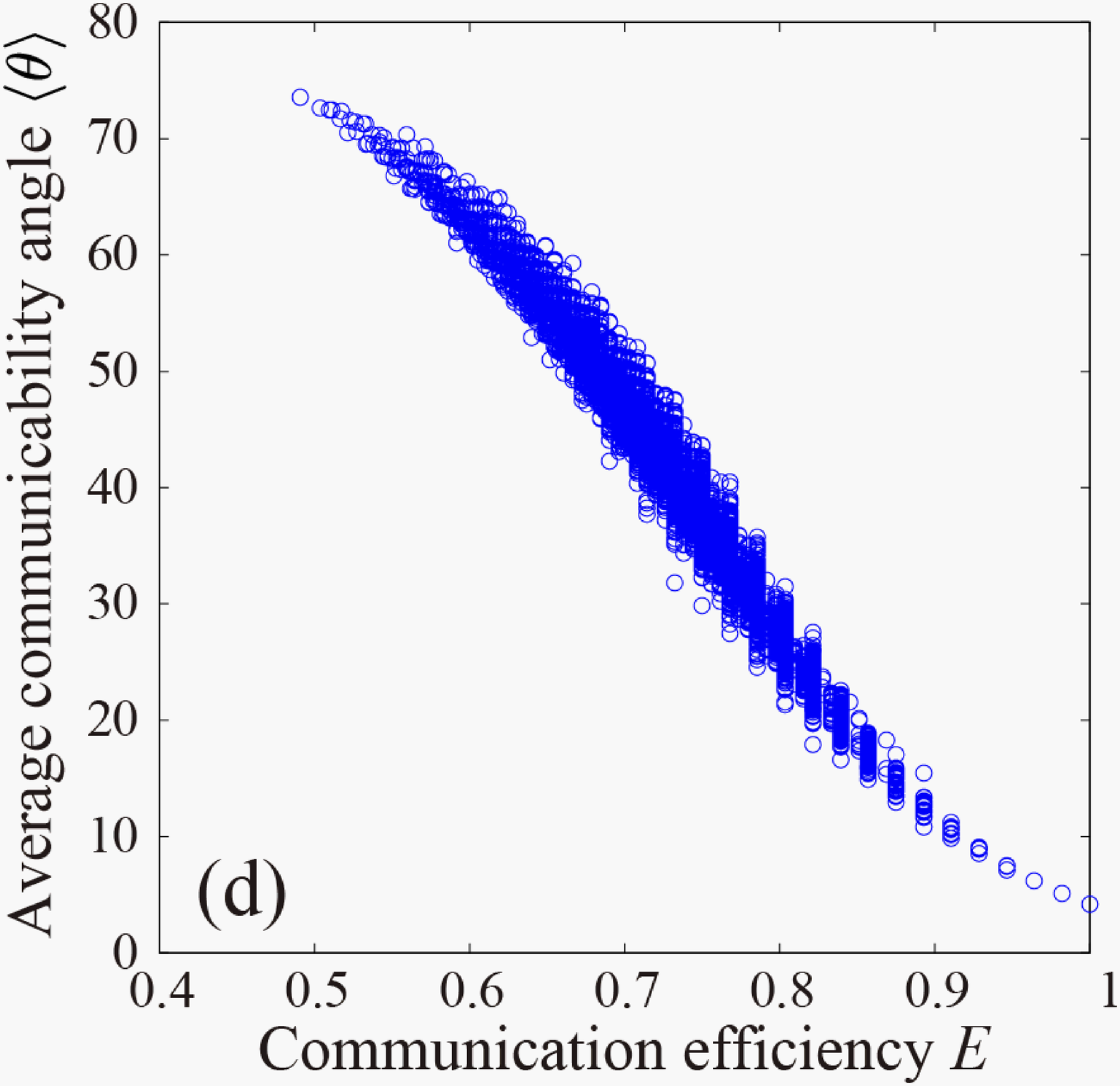} %
\hspace{0.03\textwidth} %
 \includegraphics[width=0.31\textwidth]{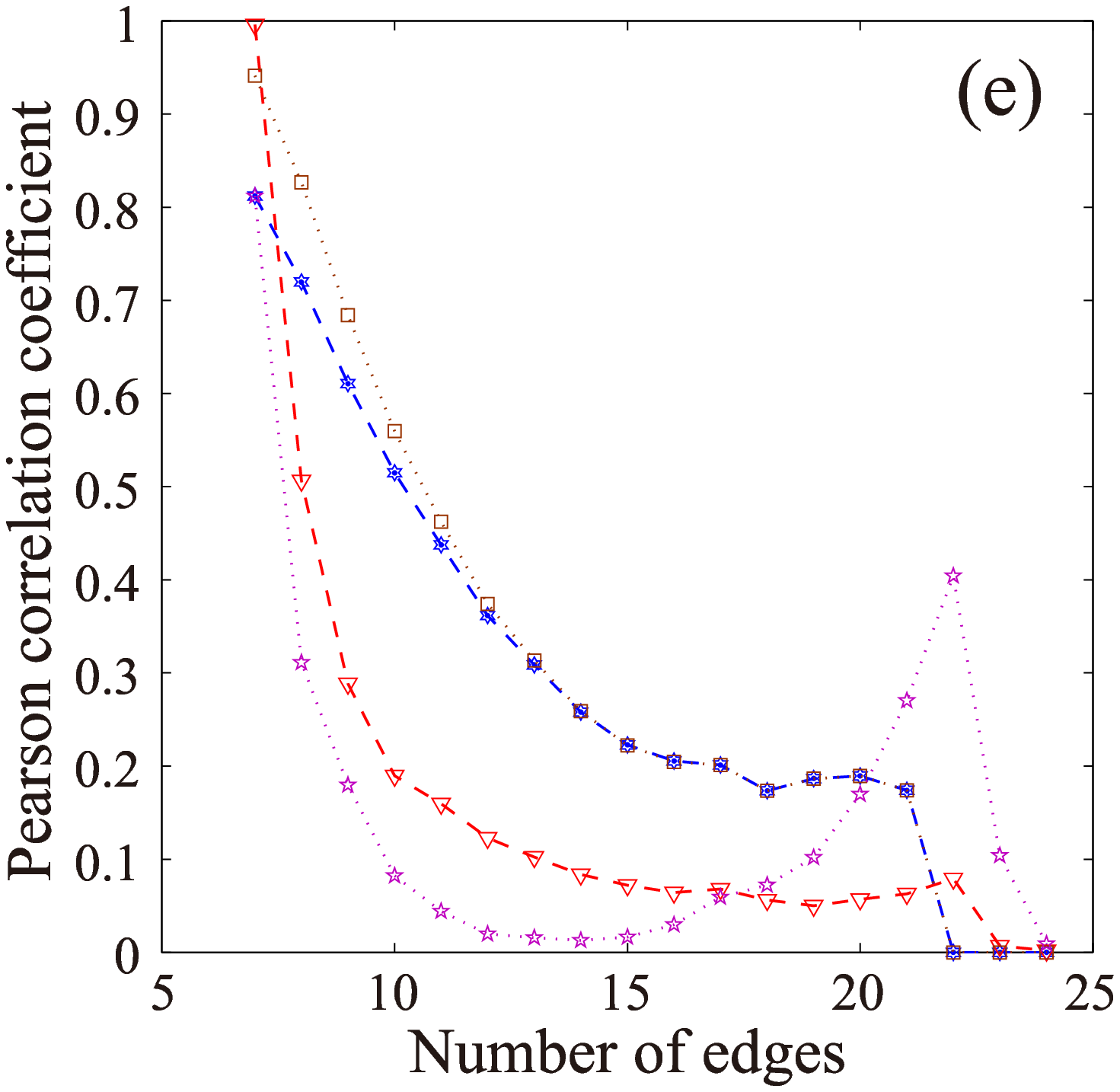} %
\hfill
\caption{(a--d) Scatter plots of the average communicability angle against (a) the
average communicability distance, (b) the average resistance distance,
(c) the average path length and (d) the communication efficiency 
for all 11,117 connected graphs with
8 nodes.
(e) The squared Pearson correlation coefficients between the metrics for 8-node graphs with a fixed number of edges.
Red triangles connected by a broken line indicate the data from (a), pink stars connected by a dotted line indicate the data from (b), the blue stars connected by a broken line indicate the data from (c) and brown squares connected by a dotted line indicate the data from (d).}
\label{fig3} 
\end{figure}

We can see that the communicability angle is not directly or trivially
related to the other metrics. It is particularly interesting to see
the lack of correlation between $\left\langle \theta\right\rangle $
and $\left\langle \xi\right\rangle $. They are highly uncorrelated 
although the two quantities are based on the same concept of communicability.
This lack of correlation is not unexpected if we consider how the
two measures and the communicability function are related to
each other via Eq.~\eqref{eq4.3}. The average communicability angle
shows more similar trends to the average path length $\left\langle l\right\rangle $,
the average resistance distance $\left\langle \Omega\right\rangle $
and the communication efficiency $E$. The extreme values of these
three measures coincide with those of $\left\langle \theta\right\rangle $,
although there is a large dispersion in between. The general plots
in Fig.~\ref{fig3} really hides the true lack of correlation that
exists among these metrics and the communicability angle. To reveal
more of these lack of correlations we plot the squared Pearson correlation
coefficient between each metric and $\left\langle \theta\right\rangle $ for
groups of graphs having the same number of edges. As can be seen in
Fig.~\ref{fig3}(e), as soon as the number of edges increases, the
correlation between the pair of indices drops significantly. For instance,
let us consider the communication efficiency, for which the correlation
with $\left\langle \theta\right\rangle $ yields a correlation coefficient
$r^{2}\approx0.94$ for the 8-node trees. This correlation coefficient
drops to $r^{2}\approx0.31$ for graphs having 13 edges and to $r^{2}\approx0.17$
for graphs having 18 edges. It is virtually zero for graphs with more
than 22 edges. The reason for this decay in the correlation is very
important. Trees have very large correlations between the pairs of
measures. This is due to the fact that in these graphs there are
only shortest paths to connect any pair of nodes because of the absence
of any cycles. As the number of edges increases the number of potential
routes between any pair of nodes increases dramatically, making more
different the measures based on shortest paths from the communicability
angle. There is also a complete lack of correlation between the communicability
angle and the average resistance distance for graphs having 10 to 20
nodes. The correlation coefficient increases for these two measures
when the number of edges is 23 but then decays. The reason for
this increase is not clear at all, but in any case the correlation
coefficient indicates that the variance in one of the indices explained
by the other is only 40\% at this point.

Among all the connected
graphs with 8 nodes, the path graph $P_{8}$ has the largest average
communicability angle and the complete graph $K_{8}$ has the smallest.
Among all the trees with 8 nodes, the star graph $K_{1,7}$ has the
smallest average communicability angle. This is also verified for
all connected graphs with 5, 6 and 7 nodes. We thereby have the following:
\newtheorem{conjecture}[theorem]{Conjecture}
\begin{conjecture} 
Among all connected graphs with $n$ nodes, the
average communicability angle is the largest for the path graph $P_{n}$
and the smallest for the complete graph $K_{n}$. 
\end{conjecture}
\begin{conjecture} 
Among all trees with $n$ nodes, the average communicability
angle is the largest for the path graph $P_{n}$ and the smallest
for the star graph $K_{1,n-1}$. 
\end{conjecture}

These observations indicate that the average communicability angle
describes the efficiency of a graph in using the space in which it
is embedded. The path graph $P_{n}$, which intuitively occupies the
largest portion of space, has the largest average communicability
angle, while the star and complete graphs, which intuitively occupy
the smallest, have the average communicability angle close to zero.
In the next section we explore more observations of this sort from
a computational point of view.

\subsection{Communicability angle and graph planarity}

Here we investigate the relation between the graph planarity and the
average communicability angle. We first determine whether a graph
is planar or not using the planarity test proposed
by Boyer and Myrvold~\cite{Boyer Myrvold}. We then construct the
histogram of the frequency of planar/nonplanar graphs with respect
to the average communicability angle.

Let $\eta_{k}$ be the number of planar graphs having $k\leq\left\langle \theta\right\rangle <\left(k+10^{\circ}\right)$
for $k=0^{\circ},10^{\circ},20^{\circ},\cdots,80^{\circ}$. We plot
in Fig.~\ref{fig4}(a) the histogram of the planar/nonplanar graphs
as a function of their values of $\left\langle \theta\right\rangle $
for all connected graphs with 8 nodes.
\begin{figure}
\includegraphics[width=0.31\textwidth]{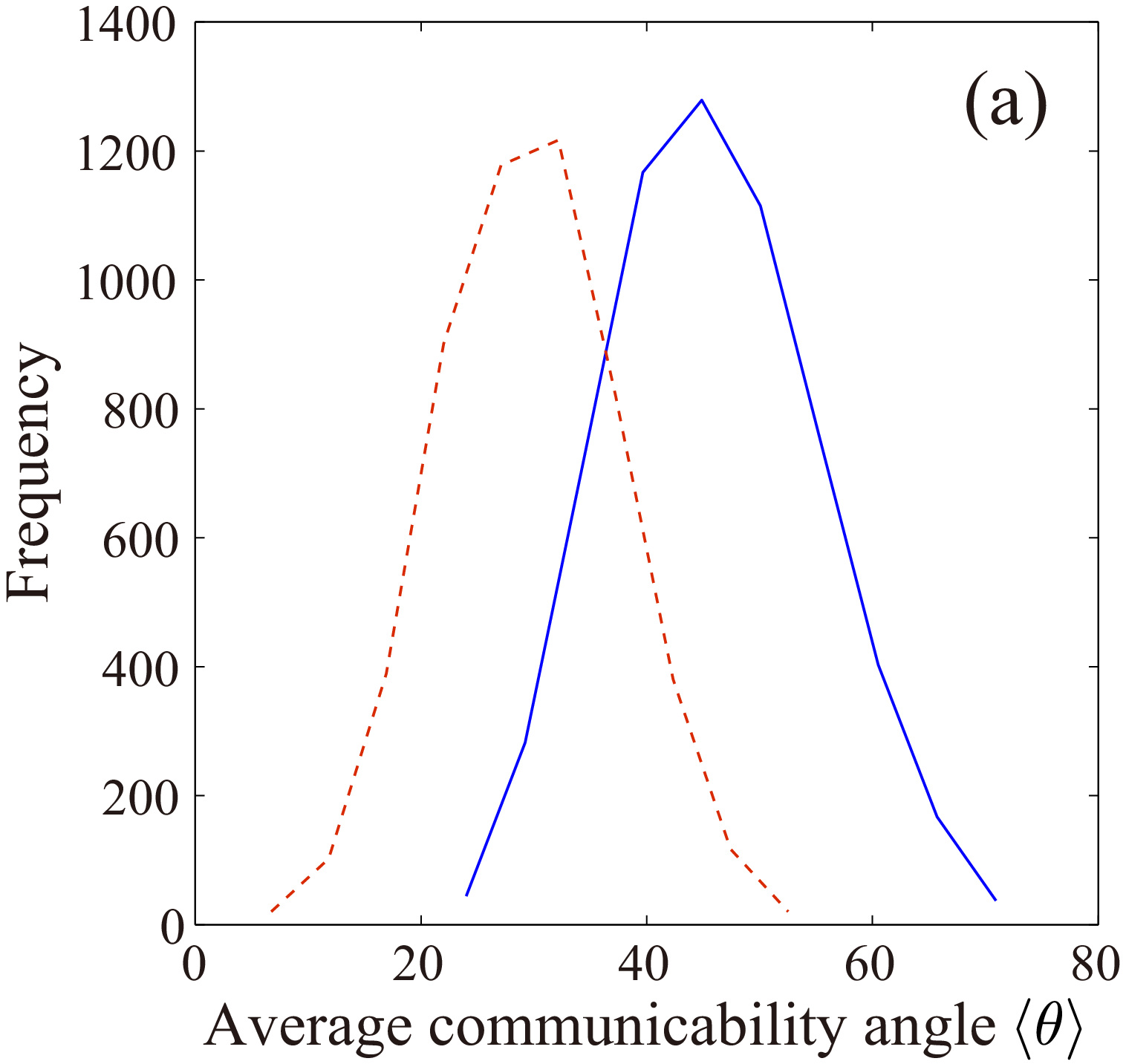} 
\hspace{0.02\textwidth}
\includegraphics[width=0.31\textwidth]{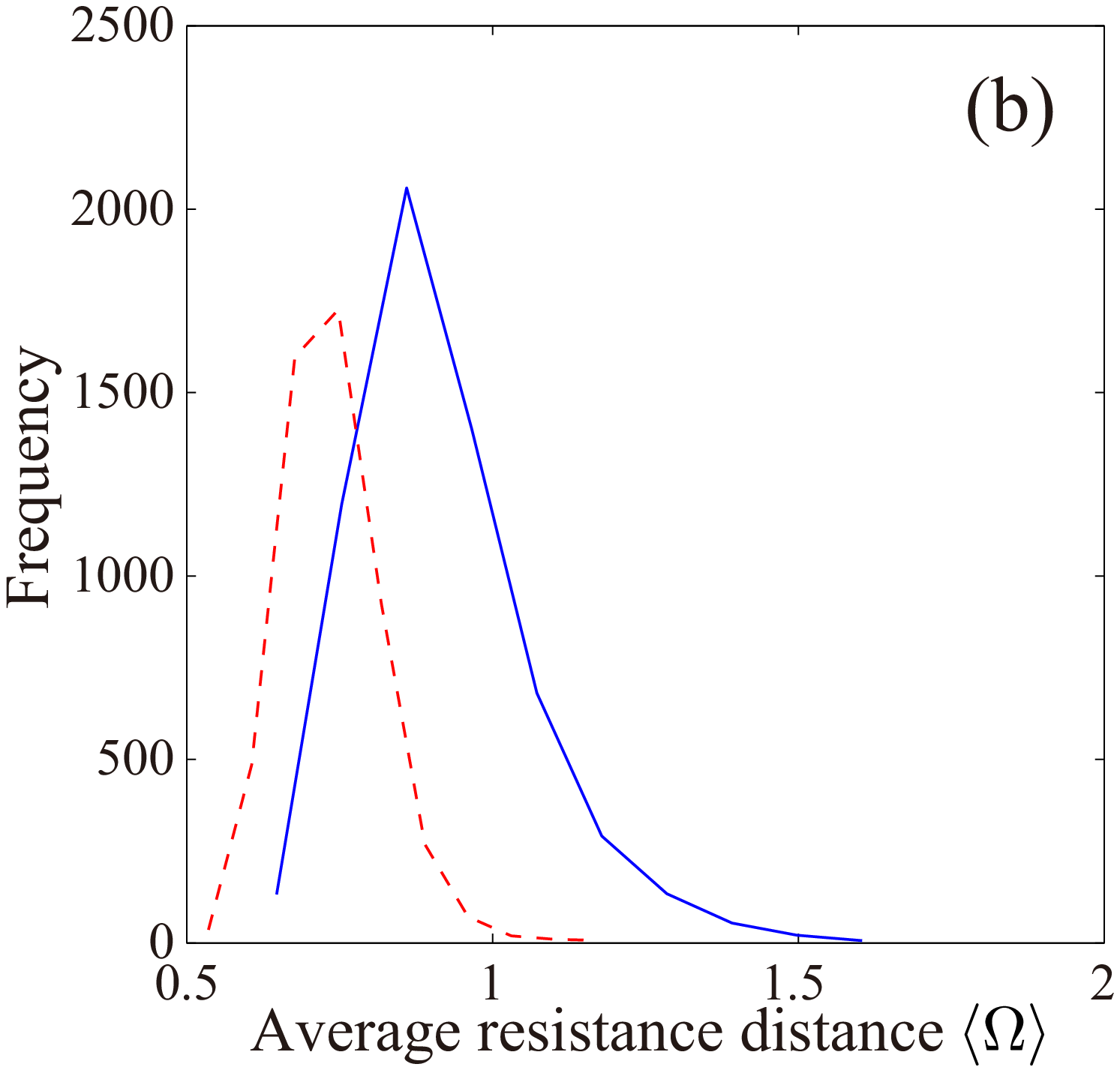}
\hspace{0.02\textwidth}
\includegraphics[width=0.31\textwidth]{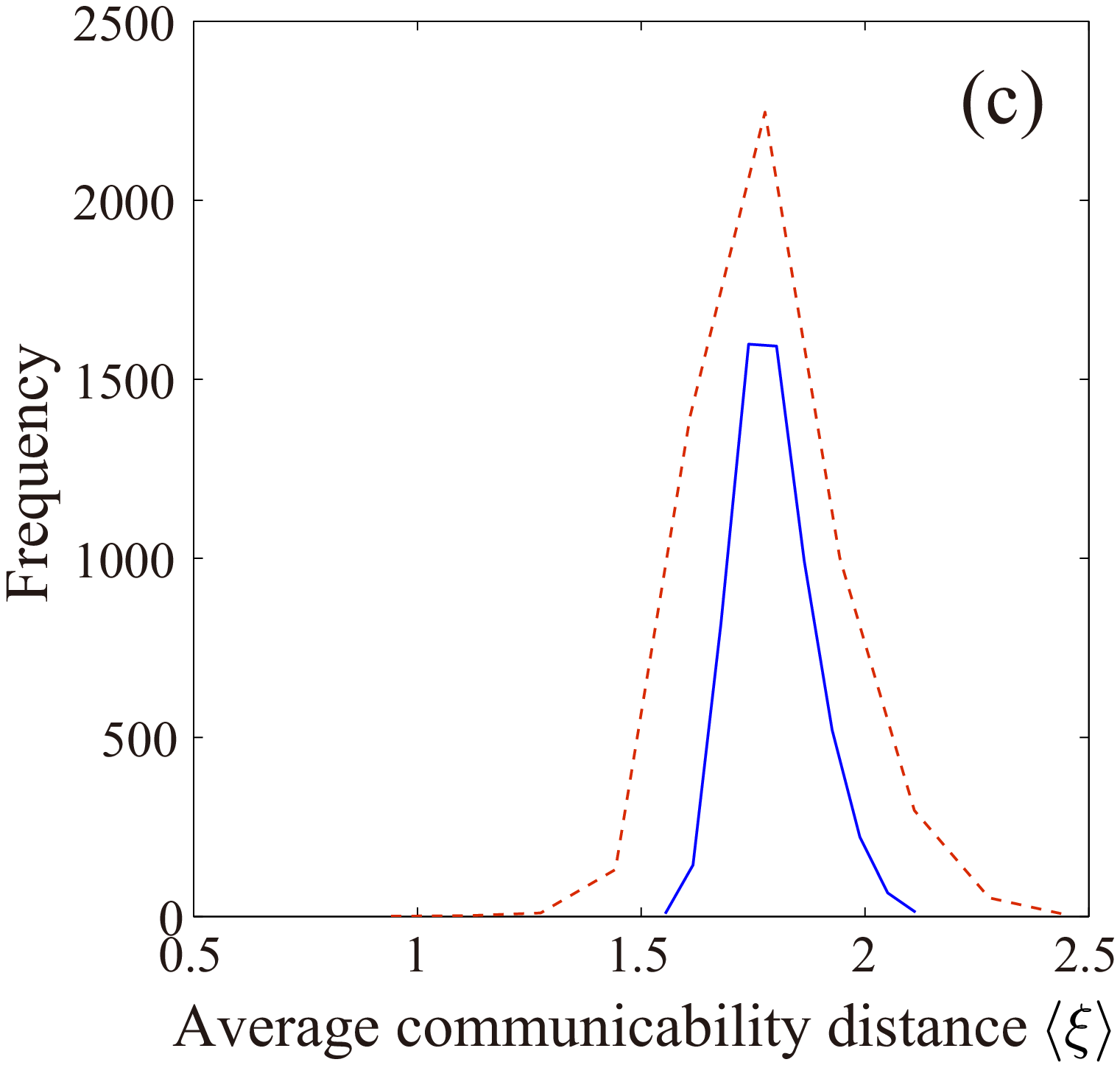}
\vspace{\baselineskip}
\\
\includegraphics[width=0.31\textwidth]{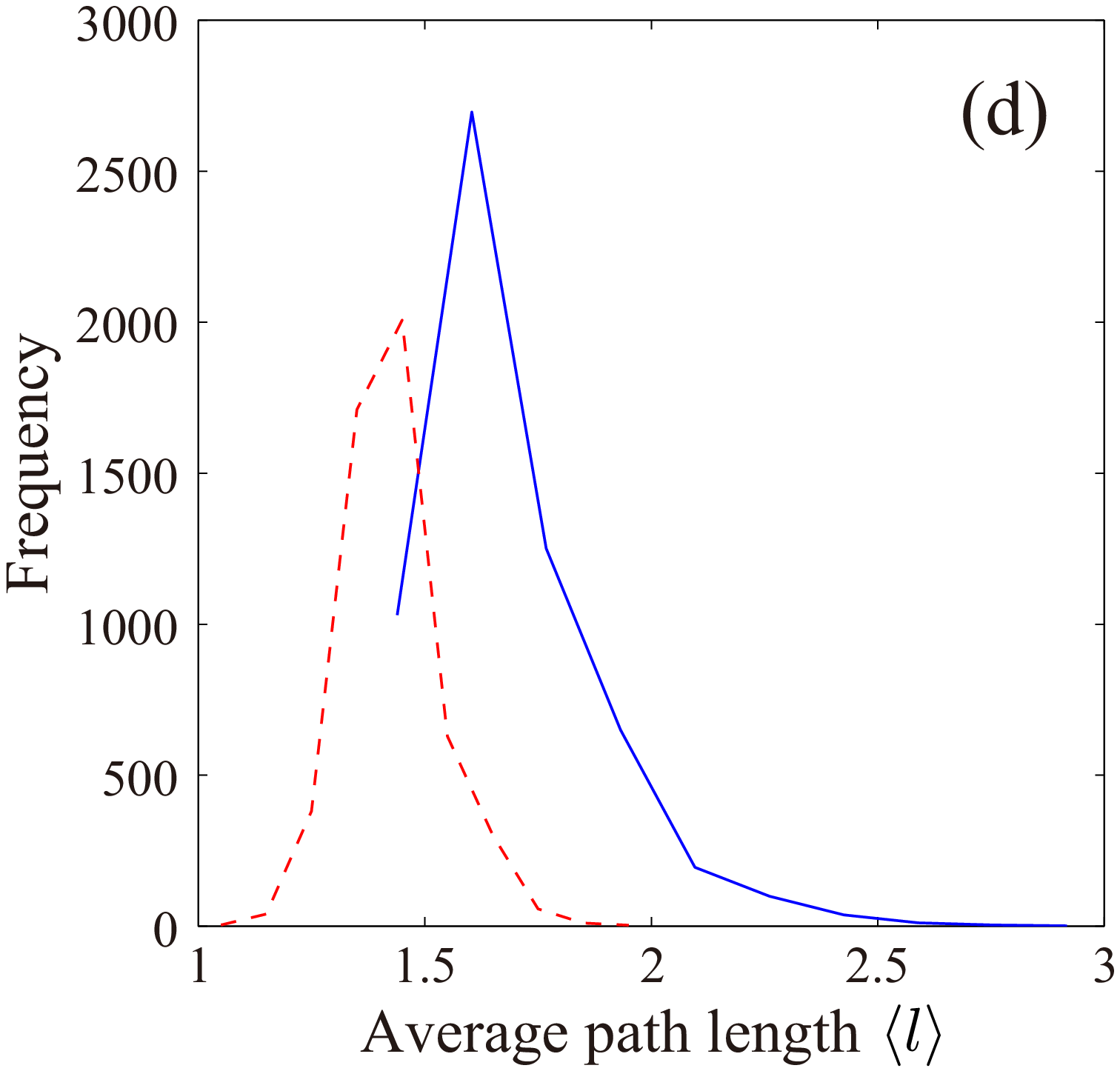}
\hspace{0.02\textwidth}
\includegraphics[width=0.31\textwidth]{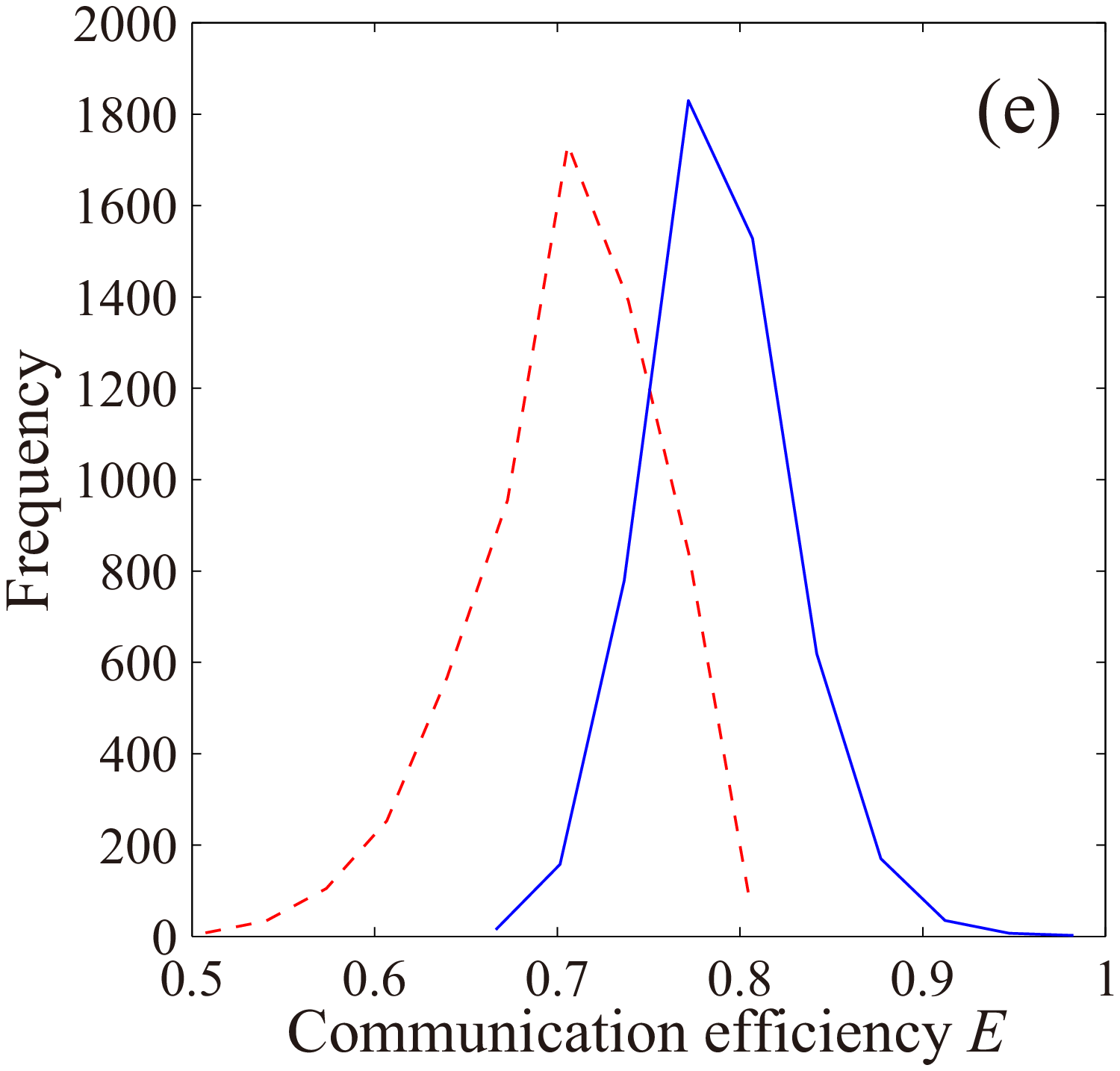}
\hfill
\caption{Frequency of planar and nonplanar graphs for different values of 
(a) the average communicability angle, (b) the average resistance
distance, (c) the average communicability distance, (d) the average
path length, and (e) the communication efficiency. The histogram for planar graphs is shown
as a solid line and that for nonplanar graphs as a broken line.}
\label{fig4} 
\end{figure}
For comparison, we also show similar plots in Fig.~\ref{fig4}(b--d) for the
average resistance distance $\left\langle\Omega\right\rangle$, the communicability distance $\left\langle\xi\right\rangle$ and 
the average path length $\left\langle l\right\rangle$. 

The first interesting observation is that the planar graphs yield
significantly larger values of $\left\langle \theta\right\rangle $
than the nonplanar graphs. The peaks in the histogram Fig.~\ref{fig4}(a)
for the planar and nonplanar graphs are at $\left\langle \theta\right\rangle \approx44.87^{\circ}$
and $\left\langle \theta\right\rangle \approx32.17^{\circ}$, respectively.
There is a larger relative separation between the two peaks of the
histogram for $\left\langle \theta\right\rangle $ than for the rest
of the measures.
Let us put this in a quantitative context. Let us define
the percentage of variation between the maxima of the two peaks as:
$v\left(\%\right)=100\times\left(x_{h}\left(\mathrm{planar}\right)-x_{h}\left(\mathrm{nonplanar}\right)\right)/\left(x_\mathrm{max}-x_\mathrm{min}\right)$,
where $x_{h}\left(\cdots\right)$ is the value of the corresponding
variable for the peak in the histogram, while and $x_\mathrm{max}$ and $x_\mathrm{min}$ are
the maximum and minimum values, respectively, of this variable $x$ for the whole dataset
of 8-node graphs. For instance, for $x=\left\langle \theta\right\rangle $, the
values are $x_{h}\left(\mathrm{planar}\right)=44.87^{\circ}$, $x_{h}\left(\mathrm{nonplanar}\right)=32.17^{\circ}$,
$x_\mathrm{max}=73.55^{\circ}$ and $x_\mathrm{min}=4.19^{\circ}$. Then, the percentages of the variation
between the maxima of the two peaks are: $18.3\%$ for $\left\langle \theta\right\rangle $,
$13.4\%$ for $E$, 
$9.7\%$ for $\left\langle \Omega\right\rangle $, and $7.7\%$ for
$\left\langle l\right\rangle $. As it is obvious from Fig.~\ref{fig4}(c)
this percentage is zero for the communicability distance. We have
repeated these experiments by considering all the 261,080 connected
graphs with 9 nodes, and the results are as follow: $23.2\%$ for $\left\langle \theta\right\rangle $,
$15.3\%$ for $E$, $10.7\%$ for $\left\langle \Omega\right\rangle $, and $9.4\%$ for
$\left\langle l\right\rangle $. Thus, it is clear that the communicability
angle not only shows the best separation between planar and nonplanar
graphs but also has the largest increase in this separation when
increasing the number of nodes.

We can elaborate more on the relation between planarity
and the communicability angle from the analysis of the connected graphs
with 8 nodes: 
(i) No planar graph has $\left\langle \theta\right\rangle <21.4^{\circ}$;
(ii) The planar graphs with the smallest value of $\left\langle \theta\right\rangle $
correspond to the \textit{maximal planar graphs}.
A graph is maximal planar, also known as a \textit{triangulation}, if
the addition of any edge to it results in a nonplanar graph. Obviously,
these are the `least planar' of all planar graphs. Examples are given
in Fig.~\ref{fig5};
\begin{figure}
\centering %
 \includegraphics[width=0.7\textwidth]{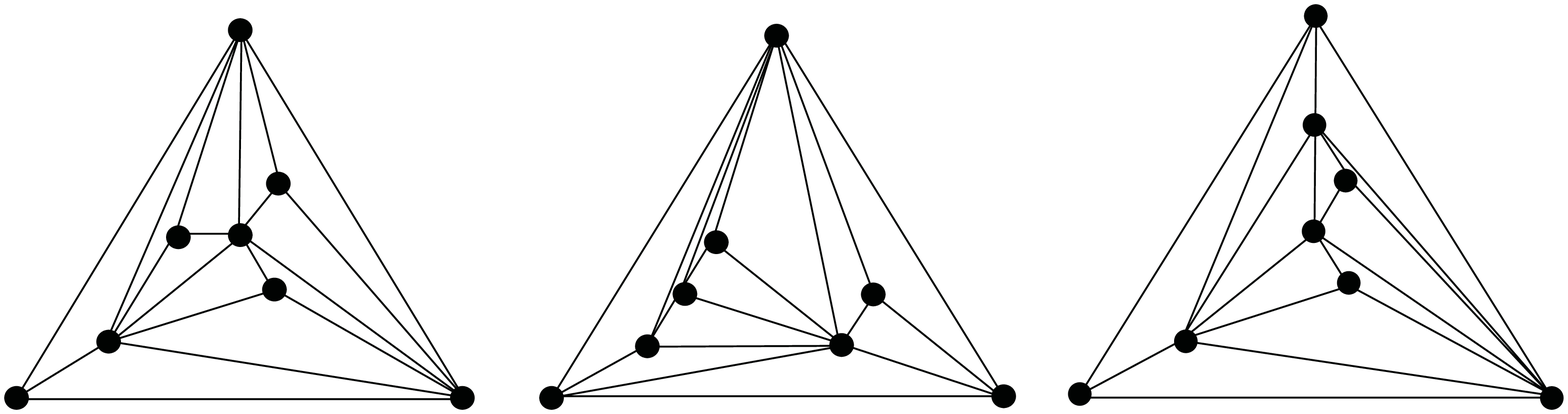}
\caption{Three maximal planar graphs with 8 nodes which have the smallest values
of $\left\langle \theta\right\rangle $. The graphs are drawn as triangulations
using Schnyder embedding~\cite{Schnyder embedding}. Because the graphs
are maximal planar, adding any edge will make the resulting graph
nonplanar.}
\label{fig5} 
\end{figure}
(iii) There is no nonplanar graph with $\left\langle \theta\right\rangle >55.065^{\circ}$;
(iv) The nonplanar graphs with the largest values of $\left\langle \theta\right\rangle $
are minimal nonplanar graphs. A minimal nonplanar graph is a nonplanar
graph for which every proper subgraph is planar, i.e., removing any
node or edge makes the graph planar. Again, these are the `least nonplanar'
of all the nonplanar graphs. Examples are given in Fig.~\ref{fig6}.
\begin{figure}
\centering %
 \includegraphics[width=0.7\textwidth]{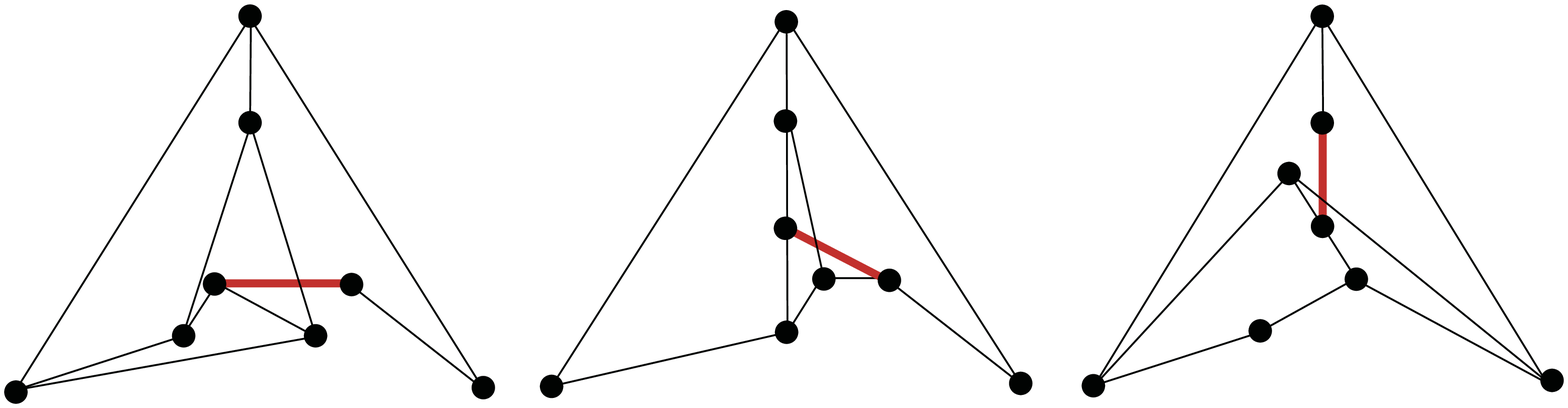}
\caption{Three minimal nonplanar graphs with 8 nodes which have the largest
values of $\left\langle \theta\right\rangle $. The graphs are drawn
using the Schnyder embedding~\cite{Schnyder embedding} and allowing
the superposition of one edge (marked in red thick line), whose removal
will transform the graph into a planar one.}
\label{fig6} 
\end{figure}

The previous results do not necessarily mean that the average communicability
angle characterizes the graph planarity or vice versa, but that the
planarity is indeed an important ingredient of the communication efficiency
as measured by the communicability angle.

\subsection{Communicability angle and graph modularity}

Modularity is a very important concept for the study of real-world
networks. It refers to the property of graphs with clusters of highly
interconnected nodes but with poor inter-cluster connectivity. Such
clusters are usually referred to as communities in network theory
and are expected to play fundamental organizational roles in real-world
networks, \textit{e.g.}, groups of proteins with similar actions and
groups of people with common interests. 

As a first example we construct
random modular graphs in the following way. We generate random modular
graphs with 1000 nodes and 50 modules. Then, with a fixed total edges
density we systematically increases the proportion of edges within
modules compared to edges across modules. As this proportion of intra-modular
edges to inter-modular edges increases, the graphs become more modular
in the sense previously explained. In order to capture the degree of
modularity of these graphs we use the Newman modularity index~\cite{modularity}, which is
defined as
\begin{equation}
Q=\sum_{k=1}^{n_{C}}\left[\dfrac{E_{k}}{m}-\dfrac{1}{4m^{2}}\left(\sum_{j\in V_{k}}k_{j}\right)^{2}\right],
\end{equation}
where $E_{k}$ is the number of edges in the $k$th module, $n_{C}$
is the total number of modules, $m$ the total number of edges and
$k_{j}$ the node degree.

In Fig.~\ref{fig7} we illustrate the results of plotting the modularity
of the random modular graphs and the average communicability angle.
As can be seen, as the modularity tends to its maximum, the average
communicability angle tends to $90^{\circ}$, indicating the decrease
in the spatial efficiency of these graphs.
\begin{figure}
\centering
\includegraphics[width=0.35\textwidth]{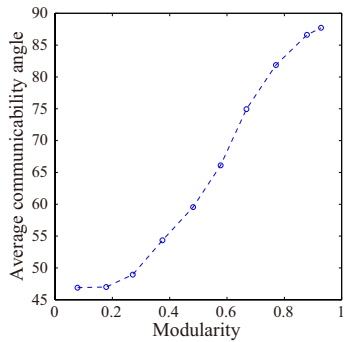}
\caption{Relation between the Newman modularity index~\cite{modularity} and the average communicability
angle for random modular graphs with 1000 nodes and 50 modules. The
total edge density is 0.01 and the proportion of intra- to inter-modular
edges varies from 0.1 to 0.95. The points in the plot indicate the average
of 100 random realizations. The broken line is to guide the eye.}
\label{fig7}
\end{figure}

A network with such clusters
has structural bottlenecks; that is, if small groups of nodes/edges
are removed the network is disconnected into two or more relatively
large connected components. An extreme case are the dumbbell graphs
$K_{n}$-$K_{n}$, that is, two cliques of $n$ nodes connected by
only one edge; the removal of the edge separates the network into
two connected components of $n/2$ nodes each.

On the other hand, a super-homogeneous graph, which is usually referred to as a good expansion graph, 
is characterized by the fact that every subset
$S$ with more than $n/2$ nodes has a large boundary, which is the
number of edges with one node inside the set $S$ and the other in
$\overline{S}$~\cite{Sarnak}. Expander graphs are characterized
by having a large spectral gap $\lambda_{1}-\lambda_{2}$ of the adjacency
matrix~\cite{AlonMilman}; see Refs.~\cite{Expander graphs,Expanders_2}
for details.

What is important for the present subsection is that expanders are
characterized by the lack of modularity, \textit{i.e.}, the lack of
tightly connected clusters which are poorly interconnected by structural
bottlenecks. In networks where $\lambda_{1}\gg\lambda_{2}$, we have
the following expression for the communicability angle: 
\begin{align}
\cos\theta_{pq}=\dfrac{G_{pq}}{\sqrt{G_{pp}G_{qq}}}\simeq\dfrac{\psi_{1}\left(p\right)\psi_{1}\left(q\right)e^{\lambda_{1}}}{\sqrt{\psi_{1}\left(p\right)^{2}e^{\lambda_{1}}\psi_{1}\left(q\right)^{2}e^{\lambda_{1}}}}=\cos0^{\circ}.
\end{align}
That is, the networks lacking any modularity are characterized by
very small value of the communicability angle. On the other hand,
in a network where $\lambda_{1}$ is not significantly larger than
$\lambda_{2}$, we make use of the expansions 
\begin{align}
G_{pp}G_{qq} & =\psi_{1}(p)^{2}\psi_{1}(q)^{2}e^{2\lambda_{1}}+\left(\psi_{1}(p)^{2}\psi_{2}(q)^{2}+\psi_{2}(p)^{2}\psi_{1}(q)^{2}\right)e^{\lambda_{1}+\lambda_{2}}\nonumber \\
 & +\psi_{2}(p)^{2}\psi_{2}(q)^{2}e^{2\lambda_{2}}+\mathrm{h.o.},
\\
{G_{pq}}^{2} & \simeq\psi_{1}(p)^{2}\psi_{1}(q)^{2}e^{2\lambda_{1}}+2\psi_{1}(p)\psi_{1}(q)\psi_{2}(p)\psi_{2}(q)e^{\lambda_{1}+\lambda_{2}}\nonumber \\
 & +\psi_{2}(p)^{2}\psi_{2}(q)^{2}e^{2\lambda_{2}}+\mathrm{h.o.},
\end{align}
where h.o.\ denotes the higher-order terms. The communicability angle
is thereby transformed into the form 
\begin{align}
\cos\theta_{pq}=\dfrac{G_{pq}}{\sqrt{{G_{pq}}^{2}+\left(\psi_{1}(p)\psi_{2}(q)-\psi_{2}(p)\psi_{1}(q)\right)^{2}e^{\lambda_{1}+\lambda_{2}}+\mathrm{h.o.}}}.\label{eq:53}
\end{align}
The second term in the denominator depends on the size of the spectral
gap; 
the closer $\lambda_{2}$ is to $\lambda_{1}$, \textit{i.e.}, the smaller
the spectral gap, the larger the denominator is, and consequently,
the smaller Eq.~\eqref{eq:53} is.
Therefore, the angle $\theta_{pq}$ is larger as the spectral gap is smaller.
We should remark here that $\theta_{pq}$ does not depend only on
the spectral gap because the higher-order terms in Eq.~\eqref{eq:53}
can make an important contribution.

Let us show examples that illustrate the above important relation
between the communicability angle and the graph modularity. Here again
we focus on $\left\langle \theta\right\rangle $. We first consider
the dumbbell graph $K_{3}$-$K_{3}$ shown in Fig.~\ref{fig8}(a).
\begin{figure}
\centering 
\includegraphics[height=0.07\textheight]{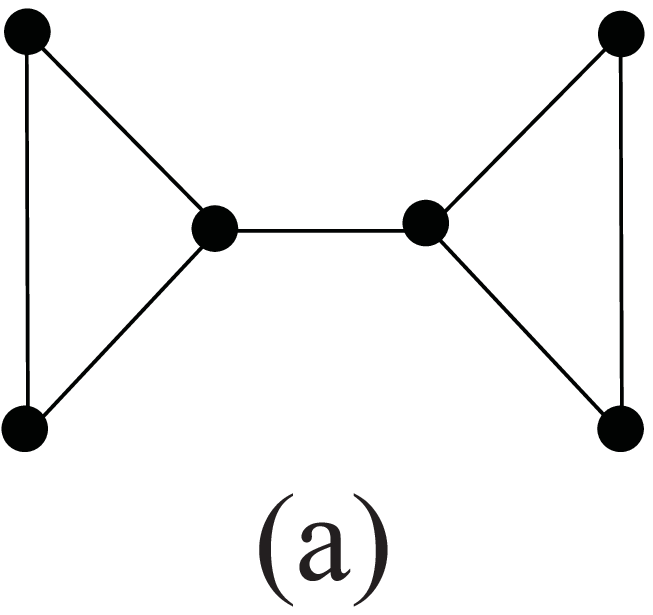} 
\hspace*{0.05\textwidth}
\includegraphics[height=0.07\textheight]{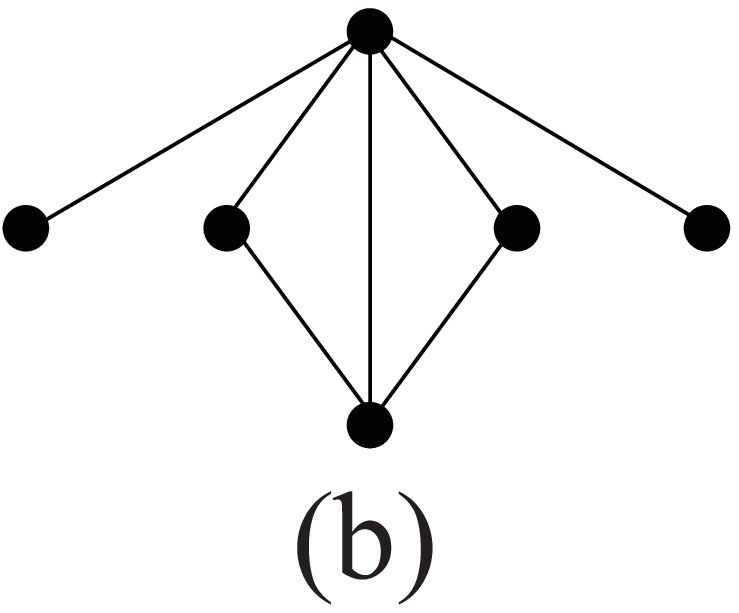} 
\hspace*{0.05\textwidth}
\includegraphics[height=0.07\textheight]{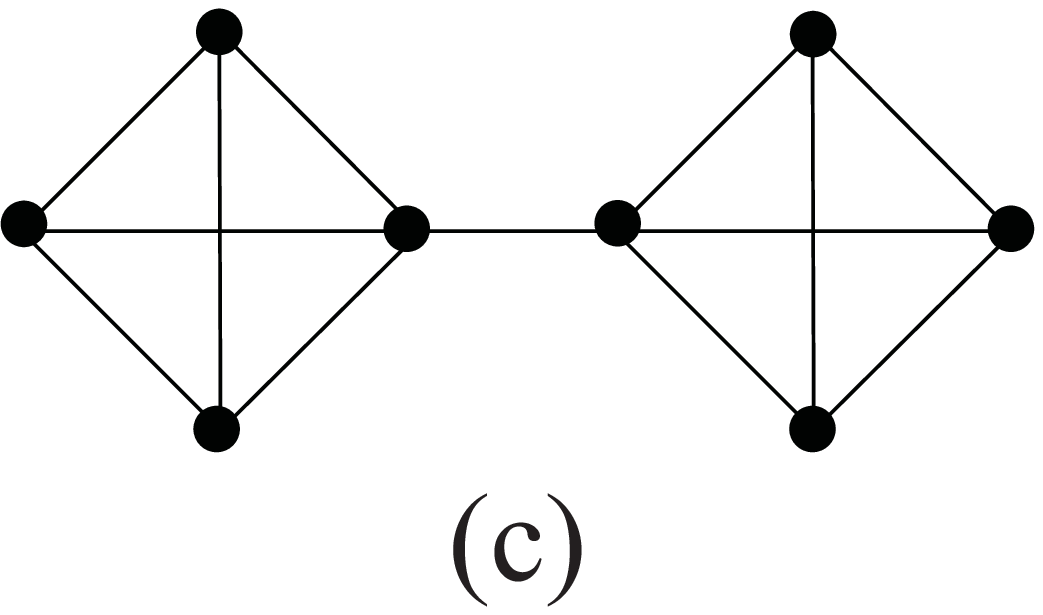} 
\hspace*{0.05\textwidth}
\includegraphics[height=0.07\textheight]{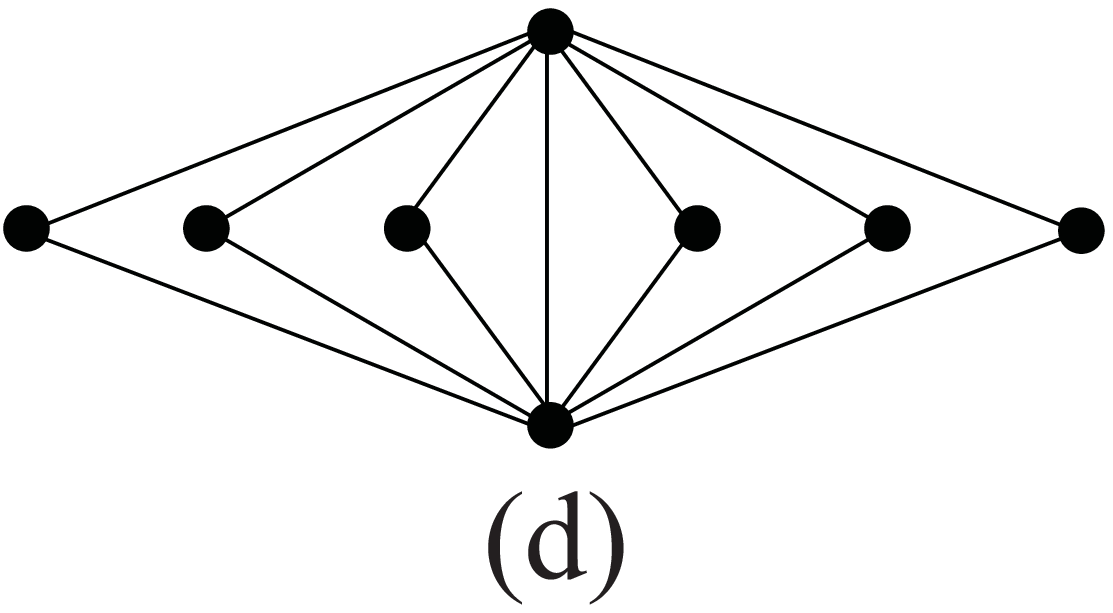} 
\caption{The graphs with 6 nodes and 7 edges (a) with the largest and (b) the
smallest average communicability angles. The same for the graphs with
8 nodes and 13 edges (c) and (d).}
\label{fig8} 
\end{figure}
It consists of two cliques of 3 nodes each, which are connected by
a link, thus having 7 edges in total. The average communicability
angle for this graph is $\left\langle \theta\right\rangle \approx57.105$
and its spectral gap is $\Delta\approx0.682$. Among the 19 graphs
with 6 nodes and 7 edges, the dumbbell $K_{3}$-$K_{3}$ has the largest
value of $\left\langle \theta\right\rangle $. The smallest value
of the average communicability angle is obtained for the graph in
Fig.~\ref{fig8}(b), having $\left\langle \theta\right\rangle \approx47.935$
and $\Delta\approx2.284$.

The situation is very similar for the 1,454 graphs with 8 nodes and
13 edges, among which the dumbbell graph $K_{4}$-$K_{4}$ in Fig.~\ref{fig8}(c)
has the largest average communicability angle $\left\langle \theta\right\rangle \approx53.876$
with the spectral gap $\Delta\approx0.511$. The graph with the smallest
value of $\left\langle \theta\right\rangle $ is the so-called agave
graph shown in Fig.~\ref{fig8}(d); it consists of two connected
nodes each of which is also connected to the other $n-2$ nodes that
are not connected among them. It has $\Delta=4.00$ and $\left\langle \theta\right\rangle \approx31.782$.
The graphs with the second and third smallest average communicability
angles, $\left\langle \theta\right\rangle \approx35.123$ and $\left\langle \theta\right\rangle \approx35.606$
with $\Delta\approx2.988$ and $\Delta\approx3.337$, respectively,
have structures similar to the agave graph. Notice that the agave
graph can be disconnected by removing two edges, but the remaining
principal connected component has $n-1$ nodes, while the removal
of 50\% of the edges in this graph creates a principal connected component
still containing 62.5\% of the nodes. This shows the robustness of
this graph to edge removal, a characteristic of good expander
graphs due to the lack of structural bottleneck.

Figure~\ref{fig9}(a--b) shows planar embeddings of the graphs in
Fig.~\ref{fig8}(a--b), respectively, onto triangular lattices. 
\begin{figure}
\centering %
\begin{minipage}[b]{0.15\textwidth}%
 \vspace*{0mm}
 \includegraphics[width=1\textwidth]{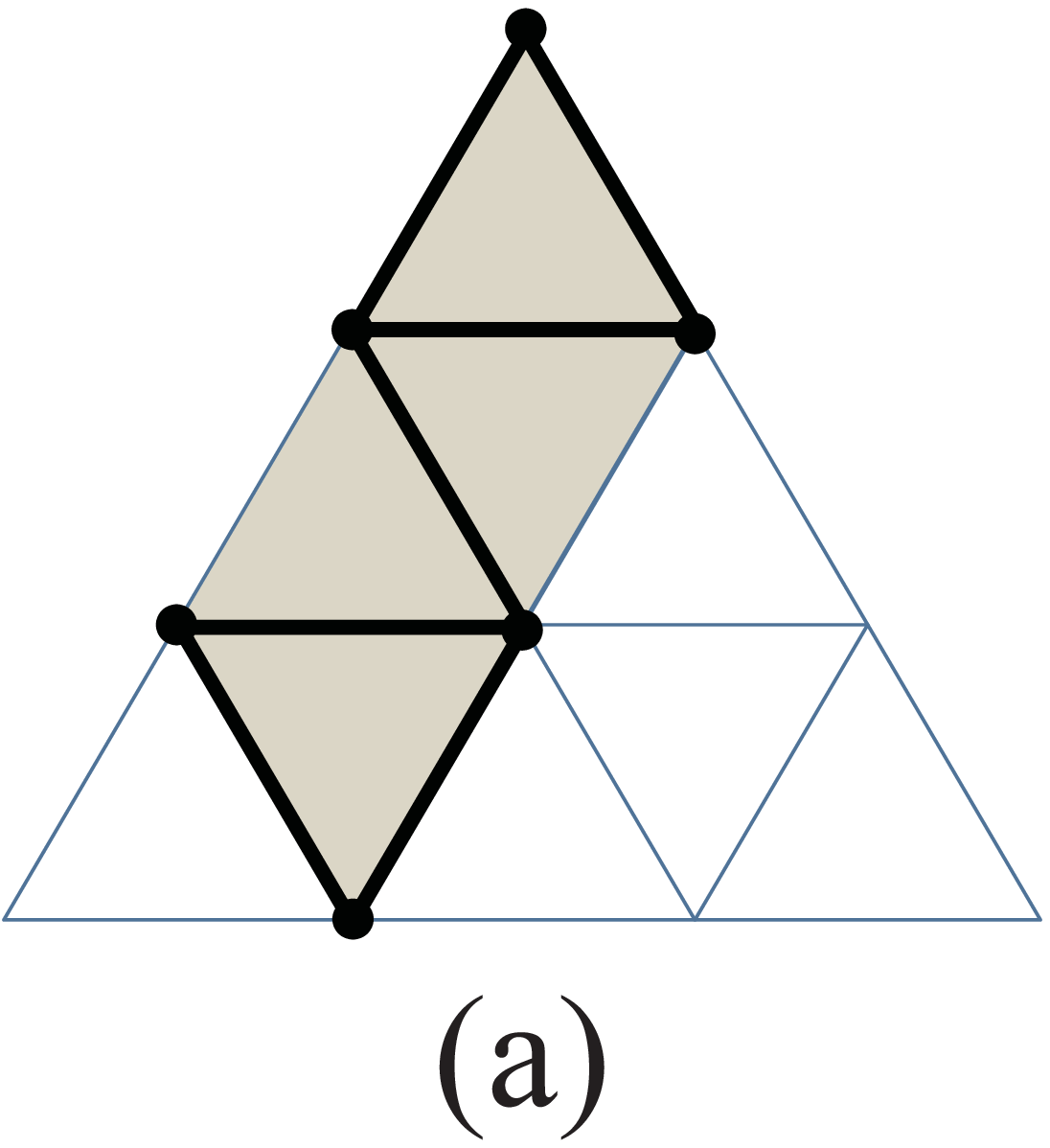} %
\end{minipage}\hspace*{0.05\textwidth} %
\begin{minipage}[b]{0.15\textwidth}%
 \vspace*{0mm}
 \includegraphics[width=1\textwidth]{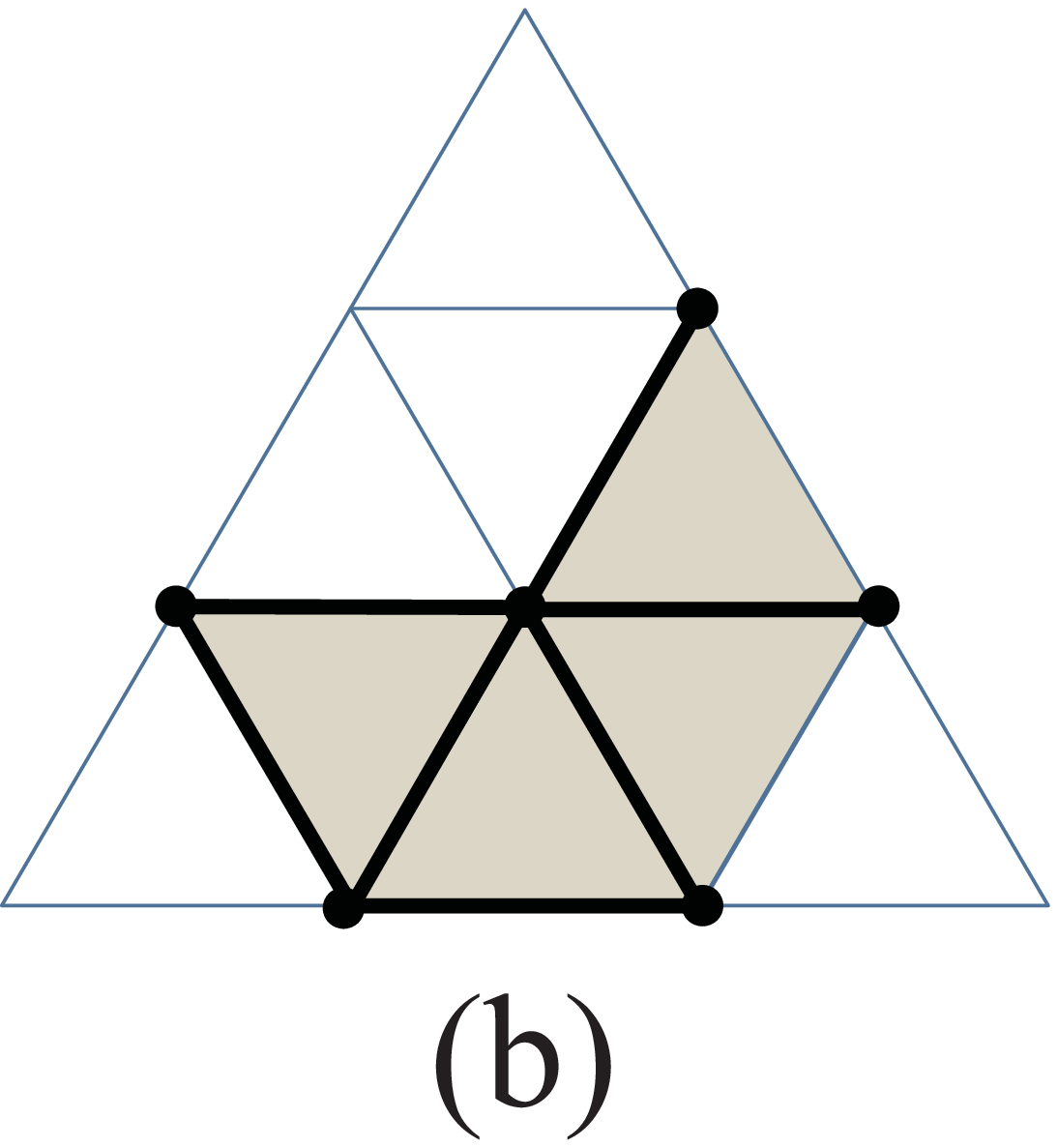} %
\end{minipage}\hspace*{0.05\textwidth} %
\begin{minipage}[b]{0.2\textwidth}%
 \vspace*{0mm}
 \includegraphics[width=1\textwidth]{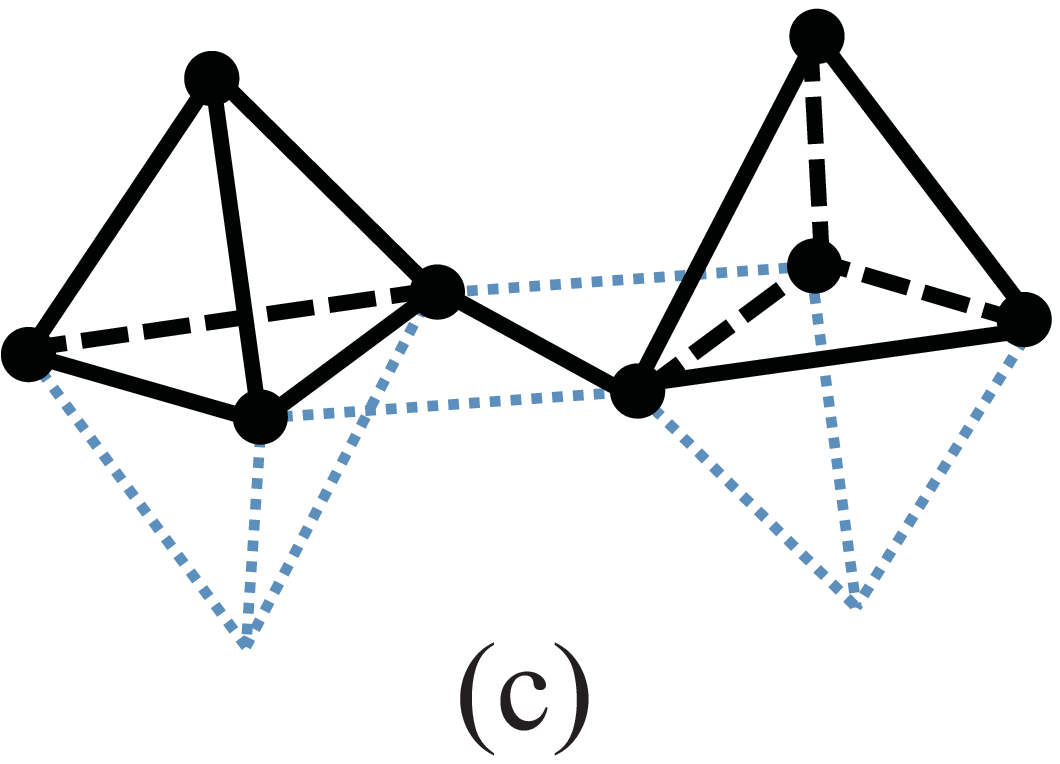} %
\end{minipage}\hspace*{0.05\textwidth} %
\begin{minipage}[b]{0.15\textwidth}%
 \vspace*{0mm}
 \includegraphics[width=1\textwidth]{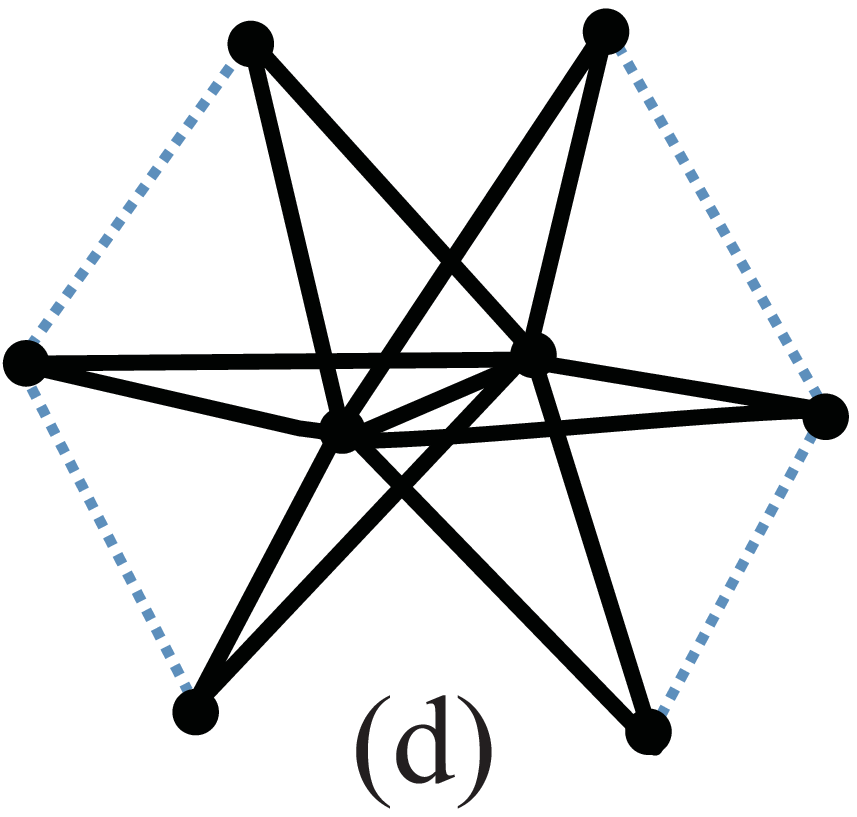} %
\end{minipage}\caption{(a--b) Planar embeddings of the graphs in Fig.~\ref{fig8}(a--b),
respectively, onto triangular lattices. (c--d) Three-dimensional embeddings
of the graphs in Fig.~\ref{fig8}(c--d), respectively, onto close-packed
lattices.}
\label{fig9} 
\end{figure}
The shadowed areas indicate the triangles covered by the graphs in
these embeddings. Although both cover the four triangles, the latter
graph, the one with the smallest average communicability angle, covers
the most efficient packing in two-dimensional space, which is the
area with a node surrounded by six others forming a hexagon. This
is known as the penny-packing problem; see Ref.~\cite{PennyPacking}
for further information. The embedding of the graph with higher modularity
and the largest average communicability angle is far from this optimal
configuration.

A similar situation occurs with the graphs in Fig.~\ref{fig8}(c--d),
the ones with the largest and smallest $\left\langle \theta\right\rangle $
among those with 8 nodes and 13 edges; Fig.~\ref{fig9}(c--d) show
their embeddings onto close-packed lattices. We can conclude from
these observations that a large average communicability angle indicates
a poor spatial efficiency of the graph, while a small value of $\left\langle \theta\right\rangle $
is associated to the efficient use of space.

\subsection{Communicability angle and graph holes}

Another characteristic of spatial efficiency that
is desirable to be captured by the average communicability angle is
the existence of holes. The presence of large holes in
a graph obviously makes its spatial efficiency very poor. For instance,
let us consider a city in which all the street form annulus such that
the whole center of the city is empty. The density of streets in that
city is very small in comparison to what it is expected from the area
occupied by the whole city.

Here we propose to consider the Sierpinski graphs
as a model of simple graphs embedded in a Euclidean space such that
the density of the graphs decays with the size. By the density we
mean here the number of nodes divided by the area occupied by the
corresponding external triangle. Let us denote by
\begin{equation}
\vec{e}_{1}=\left(1,0,0\right),\:\vec{e}_{2}=\left(0,1,0\right),\:\vec{e}_{3}=\left(0,0,1\right)
\end{equation}
the canonical basis vectors of $\mathbb{R}^{3}$.
The Sierpinski graphs are generated iteratively from $G_{0}=\left(V_{0},E_{0}\right)$,
where $V_{0}=\left\{ \vec{e}_{1},\vec{e}_{2},\vec{e}_{3}\right\} $
and $E_{0}=\left\{ \left(\vec{e}_{1},\vec{e}_{2}\right),\left(\vec{e}_{2},\vec{e}_{3}\right),\left(\vec{e}_{3},\vec{e}_{1}\right)\right\} $.
Then, for $G_{k}=\left(V_{k},E_{k}\right)$ we have~\cite{Sierpinski}
\begin{align}
V_{k>0} & =\left(2^{k-1}\vec{e}_{1}+V_{k-1}\right)\cup\left(2^{k-1}\vec{e}_{2}+V_{k-1}\right)\cup\left(2^{k-1}\vec{e}_{3}+V_{k-1}\right),\\
E_{k>0} & =\left(2^{k-1}\vec{e}_{1}+E_{k-1}\right)\uplus\left(2^{k-1}\vec{e}_{2}+E_{k-1}\right)\uplus\left(2^{k-1}\vec{e}_{3}+E_{k-1}\right),
\end{align}
where $\uplus$ represents the disjoint union of sets.
We illustrate in Fig.~\ref{fig10} the Sierpinski graphs $G_{1}$,
$G_{2}$ and $G_{3}$. The total area occupied by the graph is the
area of the external triangle which have coordinates $\left(2k,0,0\right),\left(0,2k,0\right),\left(0,0,2k\right)$.
Notice that the Sierpinski graphs $G_{0}$ and $G_{1}$ do not have
any holes, $G_{2}$ has a central hole of length 6 and $G_{3}$ has
a central hole of length 12 plus 3 holes of length 6. 
As the graph grows, $G_k$ has the central hole of length $2^{k-1}\times3$ with more holes of smaller sizes, and hence
becomes more `spongy.'
\begin{figure}
\begin{center}
\includegraphics[width=0.6\textwidth]{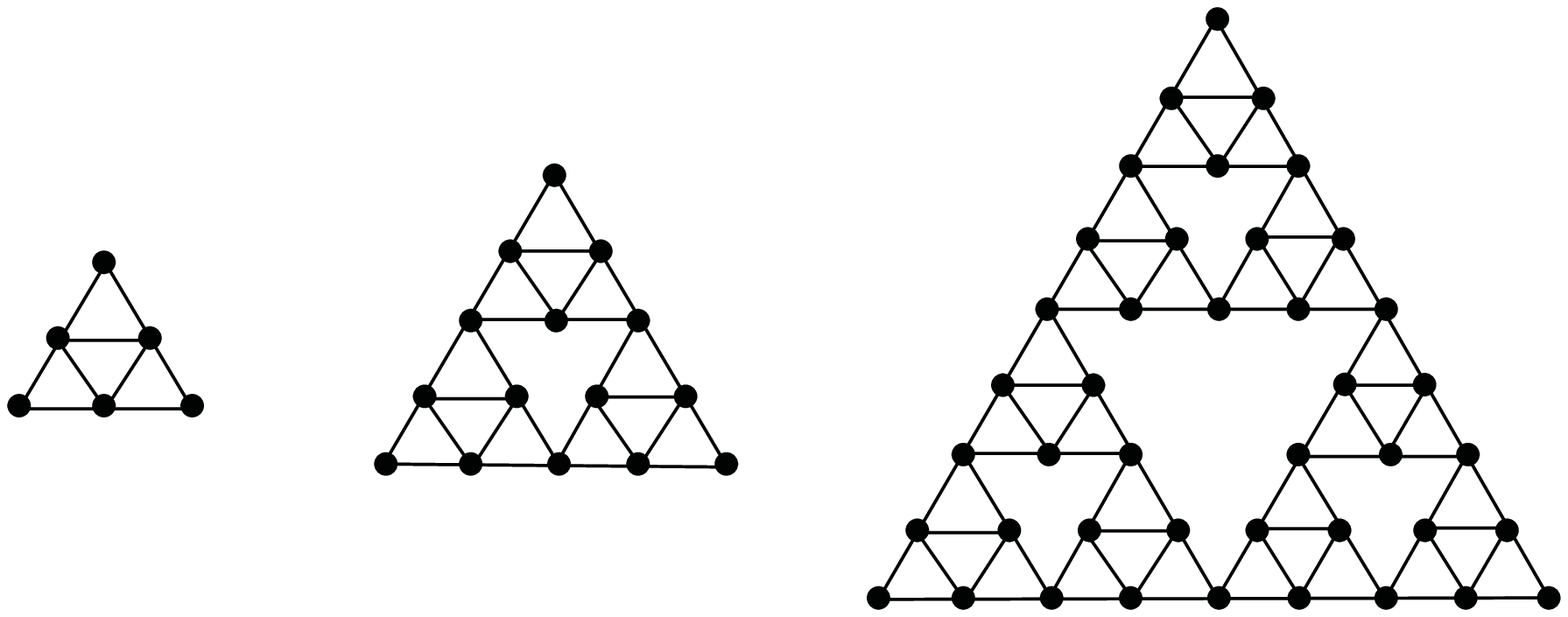}
\end{center}
 \caption{The Sierpinski graphs $G_{1},$$G_{2}$ and $G_{3}$ (from left to right).}
\label{fig10} 
\end{figure}

We have created the Sierpinski graphs for $k=1,\cdots,7$
and calculated their densities defined as the number of nodes divided
by the area of the external triangle. We illustrate in Fig.~\ref{fig11}(a) the relation between the density of the Sierpinski
graphs and the average communicability angle. 
\begin{figure}
\centering
\includegraphics[width=0.35\textwidth]{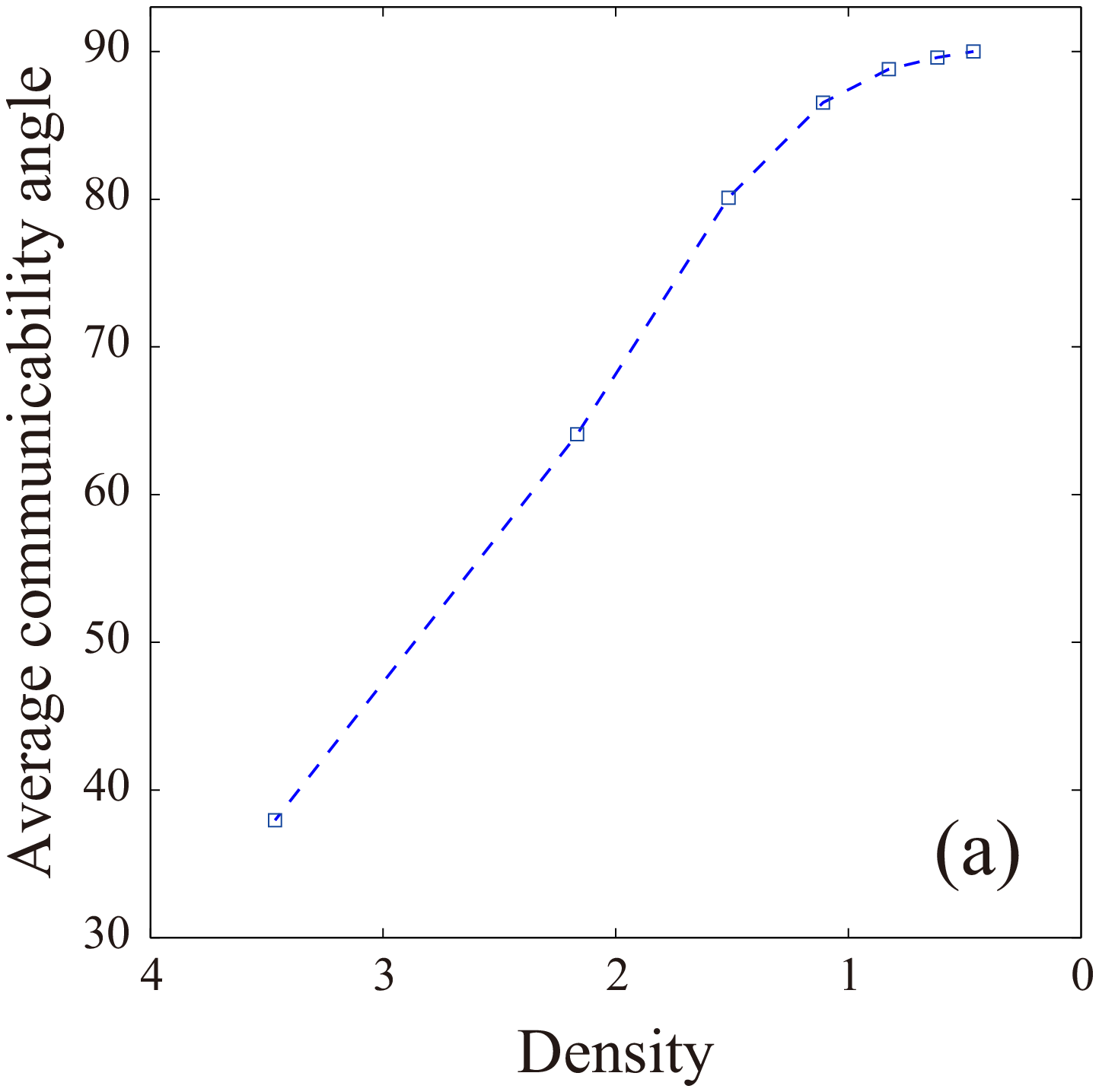}
\hspace{0.05\textwidth}
\includegraphics[width=0.35\textwidth]{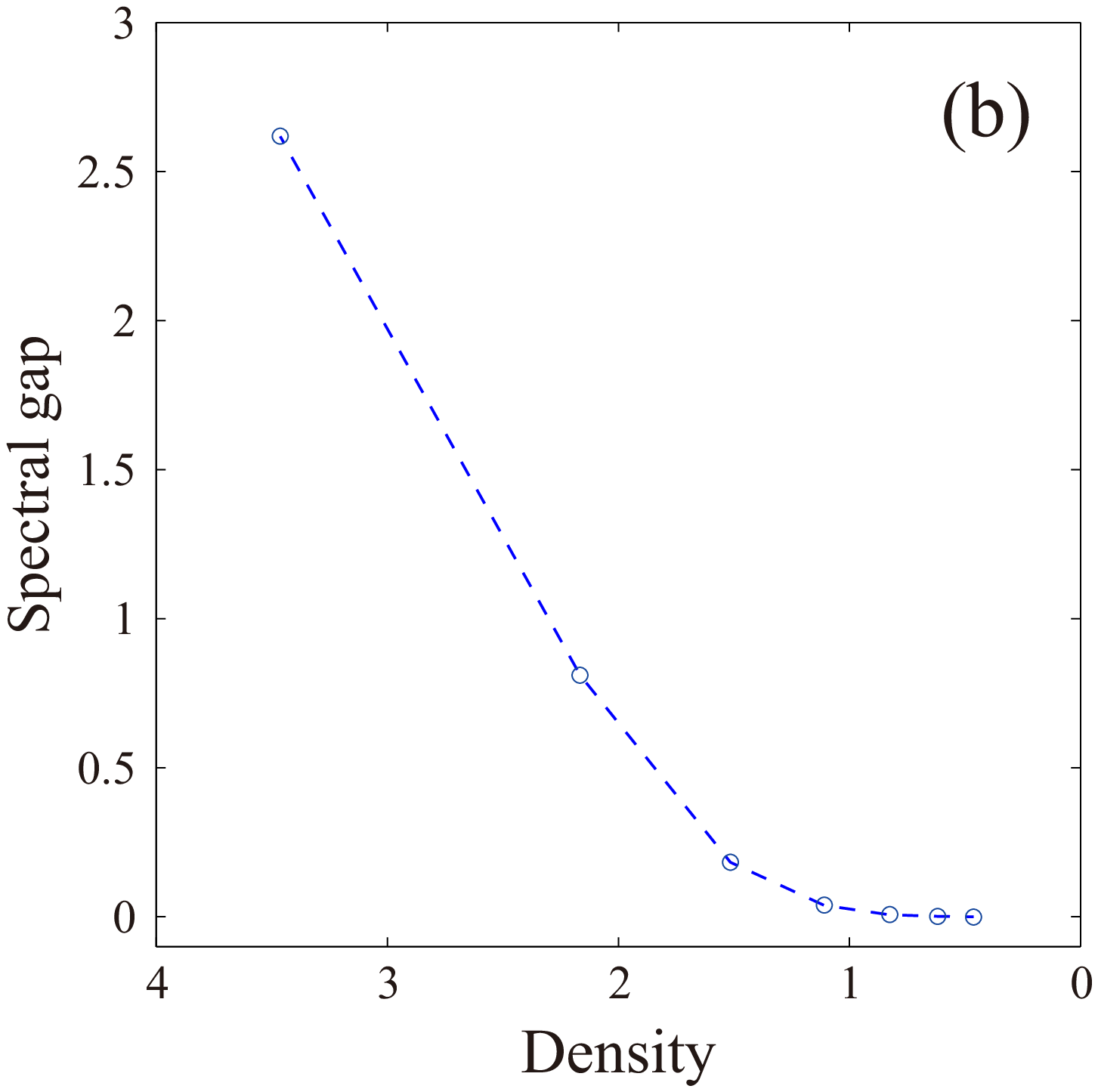}
\caption{Plots of (a) the variation of the average communicability angle
and (b) the spectral gap of the adjacency matrix with the density
of the Sierpinski graphs. The broken lines here are to guide the eye.
Notice that the $x$ axis has a reverse scale.}
\label{fig11}
\end{figure}
For $G_{1}$, which
contains no hole, the communicability angle is $\left\langle \theta\right\rangle \approx37.96^{\circ}$
although the graph is planar. As the size of the graphs increases
the average communicability angle quickly goes to its maximum for
simple graphs, \textit{e.g.}\ $\left\langle \theta\right\rangle \approx90^\circ$
for $G_{7}$, which has 3,282 nodes. 
The results illustrated in Fig.~\ref{fig11}(a)
agrees with our intuition that the communicability angle accounts
for the spatial efficiency of graphs. A Sierpinski graph with a large
number of nodes, containing very large holes, \textit{e.g.}\ the graph $G_{7}$
has a central hole of length 192, as well as many other holes of smaller
sizes, lacks spatial efficiency
in the sense of not using appropriately all the available space covered
by the external triangles. In order to understand mathematically this
relation we need to use the concept of isoperimetric number~\eqref{eq2.3}.
We recall that a graph with a small isoperimetric number contains structural
holes and/or bottlenecks, which are indications of poor spatial efficiency.
Thus, we should expect that a large Sierpinski graph has a very small
isoperimetric constant. Mohar~\cite{Mohar isoperimetric} has found
the following spectral bounds for the isoperimetric number of a graph:
\begin{equation}
\dfrac{1}{2}\left(\delta-\lambda_{2}\right)\leq i\left(G\right)\leq\sqrt{\Delta^{2}-\lambda_{1}^{2}},
\end{equation}
where $\delta$ and $\Delta$ are the minimum and
maximum degree of the graph, respectively, and $\lambda_{j}$ are the eigenvalues of
the adjacency matrix in a nonincreasing order as before. Consequently,
for graphs with bounded maximum and minimum degree --- such as the Sierpinski
graphs, where $\delta=2$ and $\Delta=4$ --- the isoperimetric number
is determined by the spectral gap $\lambda_{1}-\lambda_{2}$. A large
spectral gap indicates a large isoperimetric number, while a small
spectral gap indicates a small isoperimetric number. 
We illustrate in Fig.~\ref{fig11}(b) the plot of the density of the Sierpinski graphs against
the spectral gap of their adjacency matrices. As can be seen, the
Sierpinski graphs with small density, i.e., those with large number
of nodes, have very small spectral gaps, and consequently small isoperimetric
numbers. 

On the contrary, if a graph has a large spectral gap,
i.e., $\left(\lambda_{1}-\lambda_{2}\right)\rightarrow\infty$, the
communicability function is given by
\begin{equation}
G_{pq}\rightarrow\psi_{1,p}\psi_{1,q}\exp\left(\lambda_{1}\right),
\end{equation}
which implies that
\begin{equation}
\theta_{pq}\rightarrow0^{\circ},\forall p,q\in V.
\end{equation}
That is, a large isoperimetric number indicates that
the graphs have a large spectral gap. At the same time, a large spectral
graph indicates that the communicability angle is very small for every
pair of nodes in the graph. As we have seen a small spectral gap, and
consequently a small isoperimetric number, gives rise to a large communicability
angle as in the case of large Sierpinski graphs. This conclusion
again supports our idea that the communicability angle is a good indicator
of the spatial efficiency of a given network.

\subsection{Conclusions of the computational analysis of simple graphs}

The main conclusion of Section~\ref{sec6} is the
following: the average communicability angle very well describes a
graph characteristic which represents their spatial efficiency. It
is drawn from the following observations.  
First, planar graphs are not spatially efficient graphs; at the same time
they have large average communicability angles. On the contrary, highly
nonplanar graphs more efficiently use the available space; at the
same time they have smaller values of $\left\langle \theta\right\rangle $.
Second, a modular graph uses the available space less effectively
than a nonmodular one; at the same time, modular graphs have relatively
large values of the average communicability angle. Third,
graphs containing structural holes, which are not spatially efficient,
display large communicability angles, while those having large isoperimetric
numbers and consequently good spatial efficiency have communicability
angles close to zero.

We should, however, be careful in analyzing more complex situations
in which combinations of properties, such as nonplanarity
and modularity, or nonplanarity and structural hole, are present.
In general, we consider that graphs with relatively small values of
the average communicability angle exhibit higher spatial efficiency
than those with relatively larger values.

\section{Communicability angle in real-world networks}

We start this section by considering the average communicability angle
of a series of 120 complex networks arising from various scenarios.
The networks are briefly described in Supplementary Information accompanying
this paper, where references to the original datasets are provided. The series includes networks in which the nodes
and links are clearly embedded into geometrical spaces, such as urban
street networks, networks formed by animal nests, brain and neural
networks, protein-residue networks as well as electronic circuits
and the Internet. It also includes networks in which the nodes and
links can hardly be allocated to geographic positions, such as food
webs, social networks and software networks. The biomolecular networks
including protein-protein interaction and gene transcription networks
are also non-geographically embedded ones.

\subsection{Global properties of the communicability angle}

The 120 real-world networks studied here cover the whole spectrum
of values of the average communicability angle from $\left\langle \theta\right\rangle \approx10^{-5}{}^\circ$
for the food web of Shelf to $\left\langle \theta\right\rangle \approx89.9^\circ$
for the Power Grid network of western USA. 

The average communicability angle of these real-world networks is
not correlated to the average path length, the communication efficiency
or the resistance distance (see Supplementary Information accompanying
this paper). Just to mention an example, let us consider the network of
galleries created by ants and the collaboration network associated
with Linux open-source software system (see Supplementary Information
for details). The first network is planar due to the fact that ants
are obliged to create their corridors and galleries in a very thin
layer of sand. The second one is a highly nonplanar network. Both
networks have the communication efficiency $E\approx0.24$; according
to this index the two graphs are equally efficient in transmitting information,
something hard to believe taking into account their different topologies
and functionalities. The average communicability angle, on the other hand, clearly indicates 
the fact that the software network is highly efficient $\left\langle \theta\right\rangle \approx3.47^\circ$
while the ant network is very inefficient $\left\langle \theta\right\rangle \approx85.51^\circ$.
There are many more examples that can be extracted from the information
provided in the Supplementary Information accompanying this paper,
all of which point to the fact that the average communicability angle
is a good index to account for communication and spatial efficiency
of networks. 

The histogram in Fig.~\ref{fig12}
shows two prominent peaks at $0^\circ\leq\left\langle \theta\right\rangle \leq9^\circ$
and at $81^\circ\leq\left\langle \theta\right\rangle \leq90^\circ$. 
\begin{figure}
\centering 
\includegraphics[width=0.4\textwidth]{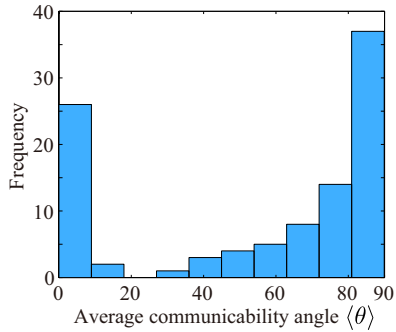} 
\caption{Histograms of the average communicability angle in 120 real-world
networks with the bin size of $9^{\circ}$.}
\label{fig12} 
\end{figure}
A more detailed view (not shown) indicates that the highest frequency
occurs at $0^\circ\leq\left\langle \theta\right\rangle \leq1^\circ$, followed
by the one at $89^\circ\leq\left\langle \theta\right\rangle \leq90^\circ$. That
is, the real-world networks are very much polarized into the two extremes;
either they have very small values of the communicability angle or
very large ones.

Certain classes of networks have a large homogeneity in the values
of the average communicability angle. The 1997 and 1998 versions of
the Internet at Autonomous System (AS) have the average communicability
angles of $0.78^\circ$ and $0.42^\circ$, respectively. There is also a large homogeneity
among the brain/neural networks, namely, the visual-cortex networks
of cat and macaque as well as the neural network of \textit{C.~elegans},
which have $\langle\!\langle\theta\rangle\!\rangle=1.77^\circ\pm1.66^\circ$,
where the double brackets $\langle\!\langle\cdots\rangle\!\rangle$ denote
the average value of the average communicability angles for a series
of networks. In addition, the classes of urban street networks formed
by 14 networks and the one of protein-residue networks formed by 40
networks also show remarkable homogeneity. For instance, the urban
street networks have $\langle\!\langle\theta\rangle\!\rangle=86.07^\circ\pm5.07^\circ$
and the protein-residue networks have $\langle\!\langle\theta\rangle\!\rangle=78.83^\circ\pm7.28^\circ$.
The ranking of the 14 cities in the former is: Barcelona $<$ Rio
Grande $<$ Yuliang $<$ Chegkan $<$ Atlanta $<$ Berlin $<$ Rotterdam
$<$ Hong Kong $<$ Mecca $<$ Cambridge $<$ Oxford $<$
Ahmedabad $<$ Milton Keynes. This means that in terms of the effective
communication among the different regions of the city, Barcelona is
the most effective one, while Milton Keynes the worse.

The homogeneity among the protein-residue networks is more unexpected
than that among the urban street networks because they represent three-dimensional
(3D) objects. Proteins are folded into 3D structures forming topologies
consisting of a mix of $\alpha$-helices and $\beta$-sheets.
They also have different shapes and sphericities. It is therefore
surprising that the protein-residue networks are characterized by
very large values of the communicability angle, which are more characteristic
of planar or almost planar networks, as demonstrated for the urban
street networks.

Although we will go back below to the relation between the communicability
angle and the structure of proteins, let us make a comment here. The
fact that proteins are embedded into the 3D physical space does not
necessarily mean that their residue networks are nonplanar. The same
applies to other naturally evolving networks, such as the networks
of galleries and corridors formed by termite mounds, which are also
characterized by very large average communicability angles with $\langle\!\langle\theta\rangle\!\rangle=88.33^\circ\pm1.01^\circ$.
Although the mounds are constructed in the 3D space, they are remarkably
close to planar graphs; we have indeed found that by removing only
6\% of the edges of these networks the graphs representing them become
planar. Both the termite mounds and the protein-residue networks have
certainly evolved in the 3D space, but the networks must be close
to planar graphs for different ecological or biological reasons. In
the termite mounds the use of a large volume of the 3D space is needed
to produce a ventilation system necessary to discharge the carbon
dioxide produced in its interior. For protein, structures close to
planar ones are needed to avoid high compactness that destroy the
internal cavities of the protein needed for developing their functions;
see Section~7.2 below.

On the other hand, the values of $\left\langle \theta\right\rangle $
obtained for the software networks~\cite{Myers Software} are unexpectedly
heterogeneous. These networks yield $\langle\!\langle\theta\rangle\!\rangle=57.6^\circ\pm30.7^\circ$
with the values ranging from $\left\langle \theta\right\rangle \approx3.465^\circ$
for Linux to $\left\langle \theta\right\rangle \approx84.323^\circ$ for
XMMS. The ranking of these networks in terms of the average communicability
angle is: Linux $<$ MySQL $<$ VTK $<$ Abi Word $<$ Digital Material
$<$ XMMS. The classes of social and biological networks consisting
of 14 and 11 networks, respectively, also show relatively large variability
in their values of the communicability angle: $\langle\!\langle\theta\rangle\!\rangle=55.8^\circ\pm21.3^\circ$,
and $\langle\!\langle\theta\rangle\!\rangle=63.3^\circ\pm17.0^\circ$, respectively.
This is not surprising; we can easily associate it to the diversity
of networks in these classes.

What is really surprising is that the food webs, which form a very
homogeneous class of networks in terms of the relations accounted
for them, yield a relatively large standard deviation in the values
of the communicability angle: $\langle\!\langle\theta\rangle\!\rangle=7.1^\circ\pm16.1^\circ$
with the values ranging from $\left\langle \theta\right\rangle \approx10^{-5}{}^\circ$
for the marine system of Shelf to $\left\langle \theta\right\rangle \approx78.356^\circ$
for the web of the English grassland. The ranking of these food webs
in terms of the average communicability angle is: Shelf $<$ Elverde
$<$ Skipwith $<$ ReefSmall $<$ LittleRock $<$ Stony $<$ Coachella
$<$ Canton $<$ Benguela $<$ BridgeBrook $<$ Ythan2 $<$ Ythan1
$<$ StMartins $<$ StMarks $<$ ScotchBroom $<$ Chesapeake $<$
Grassland.

In terms of the individual values of $\left\langle \theta\right\rangle $,
the results obtained for these 120 networks agree with our findings
in the previous section. The largest average communicability angles
are observed for the Power Grid of western USA and urban street networks,
which are planar or almost planar with both nodes and edges embedded
into a plane. On the other extreme of the smallest average communicability
angles, there are networks which are highly nonplanar, such as the
USA air transportation network, a world trade network, the Internet
at AS, and brain/neural networks. All these networks have nodes embedded
into two- or three-dimensional spaces, such as cities, countries or
organs, but the edges connecting them very efficiently use the available
space. We would like to remark here that the small values of $\left\langle \theta\right\rangle $
observed in some classes of networks do not necessarily mean a high
interconnection density. For instance, the USA airport transportation
network and the two versions of the Internet studied here have relatively
small edge densities: 0.039 and 0.0011, respectively.

\subsection{Communicability angle and spatial efficiency of proteins}

We have accumulated several pieces of empirical evidence that support
the idea that the average communicability angle accounts for the spatial
efficiency of graphs. It is, however, generally difficult to find
quantitative measures of the spatial efficiency in real-world complex
networks to compare with the communicability angle.

An exception to this is provided by proteins, which are 3D objects
characterized by different degrees of packing or spatial efficiency.
In this section we study the relation between the average communicability
angle and the spatial efficiency of the protein-residue networks for
a group of 40 proteins whose 3D structures have been resolved by X-ray
crystallography and deposited in the protein databank (PDB)~\cite{PDB}.
Here each node represents an amino acid in the protein and two nodes
are connected if the corresponding amino acids are separated at a
distance of no more than 7\AA \ in the 3D structure of the protein
as determined experimentally~\cite{Residue networks}.

A protein is a linear sequence of amino acids connected by peptide
bonds. The chain is folded into a 3D shape unique to each protein.
While the amino-acid sequence forms the so-called primary structure
of the protein, the 3D folding defines its secondary and tertiary
structures. The secondary one is characterized by the presence of
the $\alpha$-helices and the $\beta$-sheets, while the tertiary
one is formed by global positioning of the secondary one into a 3D
shape that gives the protein its globular-like structure~\cite{Protein folding}.
The folding of the proteins is the consequence, \textit{grosso modo},
of two main necessities that the protein has: (i) protecting the hydrophobic
amino acids from their contact with water; (ii) occupying a minimum
space inside the limited volume of the cell. Thus the packing of a
protein is related to its spatial efficiency~\cite{Protein packing},
which is responsible for many of its physico-chemical and biological
properties.

There are many ways of quantifying the packing of a protein, but here
we consider the following one. Let $V_{e}$ be the volume of a protein
which is expected from its ideal 3D structure and let $V_{o}$ be
the volume which is actually observed in its X-ray crystallography.
We then define the relative deviation from its ideal volume as 
\begin{align}
P=\dfrac{V_{e}-V_{o}}{V_{e}}.
\end{align}
Hereafter we call $P$ the relative packing efficiency of the protein.
A positive value of $P$ means that the protein is more packed than
expected from its ideal 3D structure, that it is highly efficient
in using the 3D space, at least relatively to the ideal structure.
A negative value of $P$, on the other hand, means that it is less
packed than expected, that it is not spatially efficient. We should
mention here that values that deviate very much from the expected
or ideal values can indicate possible problems with the structure
and as such should be discarded from the analysis.

Using computational techniques and VADAR software described in Ref.~\cite{VADAR},
we have calculated the expected and observed volumes of the 40 proteins.
We show in Fig.~\ref{fig13} the relation between the relative packing
efficiency $P$ and the average communicability angle of the 40 proteins.
\begin{figure}
\centering\includegraphics[width=0.4\textwidth]{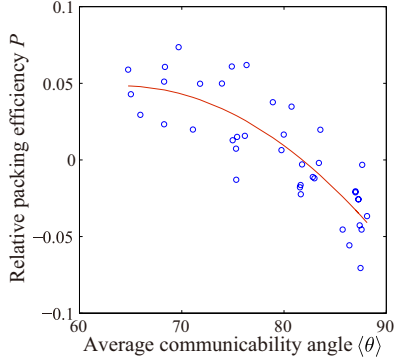} 
\caption{Linear correlation between the average communicability angle of proteins
represented by residue networks and the relative packing efficiency.}
\label{fig13} 
\end{figure}
The Pearson correlation coefficient is $R=-0.837$, indicating a significant
correlation between the two variables. We can summarize the results
as follows: 
(i) proteins with poor spatial
efficiency, $P<0$, have $\left\langle \theta\right\rangle >81^{\circ}$;
(ii) those with high spatial efficiency, $P>0$, have $\left\langle \theta\right\rangle <80^{\circ}$.
In other words, small average communicability angles
are related to high spatial efficiency of proteins while large average
communicability angles with a poor use of space. We notice in passing
that there are no proteins with $\left\langle \theta\right\rangle <60^{\circ}$,
which can be explained by the fact that increasing too much packing
would make the internal cavities of the protein disappear~\cite{Protein packing}.
The internal cavities are responsible for the interaction
of proteins with other biological molecules and usually play a fundamental
role in their functionality. In general, we can conclude
that proteins are spongy in a similar way as the Sierpinski graphs
are. 

Possibilities which the communicability angle brings to the analyses
of the structure of spatially embedded complex networks obviously
go beyond the use of $\left\langle \theta\right\rangle $. For instance,
the contour plot of the communicability angle for every pair of residues
in a protein can reveal important properties of its 3D structure.
Figure~\ref{fig14} shows an example of the protein with PDB code
1amm, which corresponds to the GammaB crystalline, whose crystallographic
analysis was carried out at 150K.
\begin{figure}
\centering %
\begin{minipage}[t]{0.4\textwidth}%
 \vspace{0mm}
 \includegraphics[width=1\textwidth]{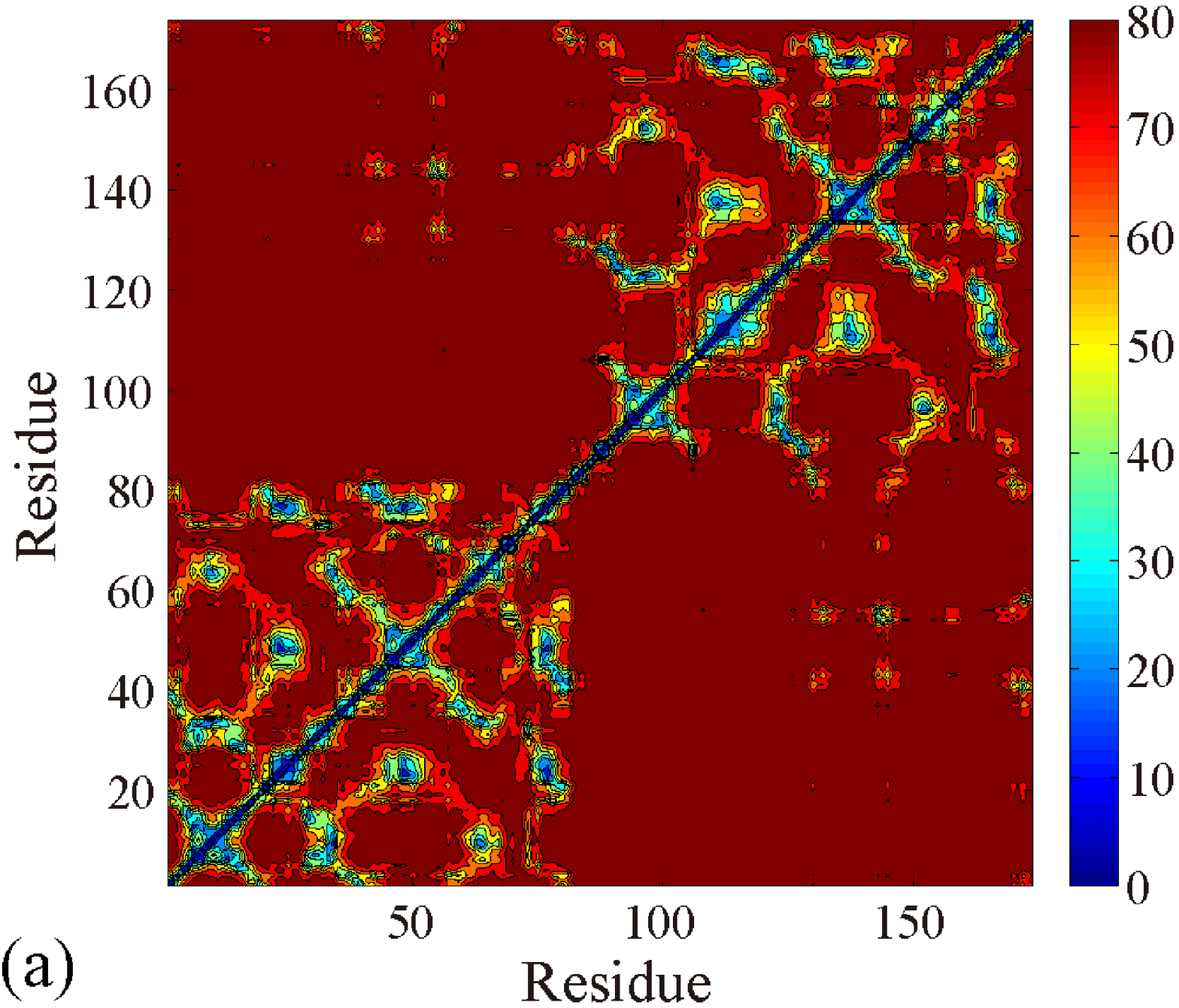} %
\end{minipage}
\hspace*{0.05\textwidth} %
\begin{minipage}[t]{0.35\textwidth}%
 \vspace{0mm}
 \includegraphics[width=1\textwidth]{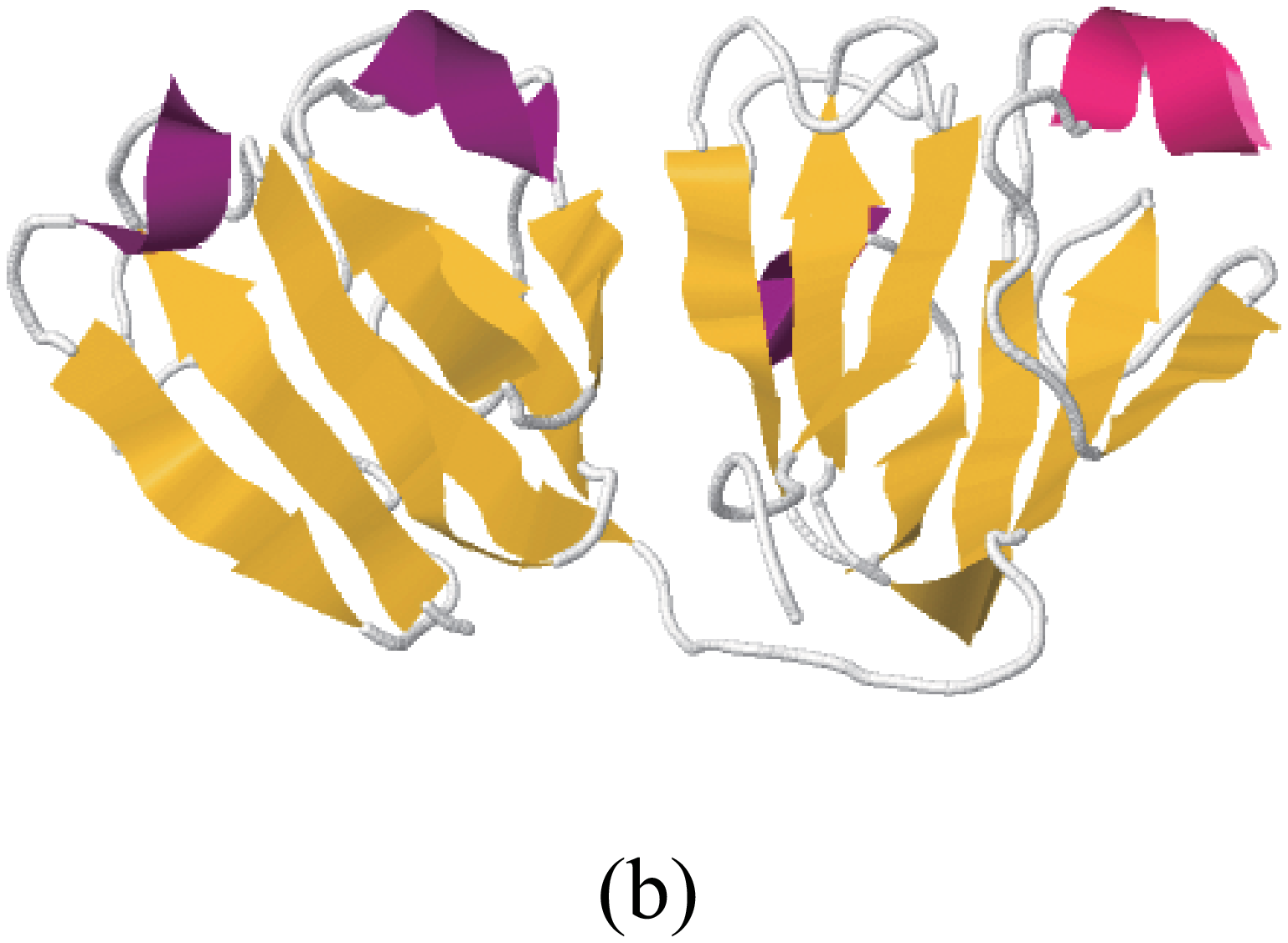} %
\end{minipage}
\caption{(a) Contour plot of the communicability angle between every pair of
residues in the GammaB crystalline protein with PDB code 1amm. (b)
A cartoon representation of the protein with PDB code 1amm in which
the $\beta$-sheets are represented as arrows in yellow and the helices
as ribbons in magenta. This shows the existence of two domains in
it.}
\label{fig14} 
\end{figure}
This protein consists of two $\alpha,\beta$-domains, the first of
which formed by amino acids 1-83 and the second by amino acids 84-174.
The two domains are very well reflected in the contour plot Fig.~\ref{fig14}(a)
as two main diagonal blocks of relatively small communicability angles,
which indicates good internal communication in each domain.

\subsection{Spatial efficiency in networks under external stress}

The communicability function has been previously generalized to consider
an external stress to which the network is submitted. This external
stress is accounted for by means of the so-called inverse temperature
$\beta\equiv\left(k_{B}T\right)^{-1}$, where $k_{B}$ is a constant
and $T$ is the temperature~\cite{Temperature}. This analogy results
from regarding that the whole network is submerged into a thermal
bath of the inverse temperature $\beta$; see~\cite{Estrada Hatano Benzi,EstradaBook}
for details. After equilibration in the bath, all edges of the network
acquire a weight equal to $\beta$.

It is clear that when $\beta\rightarrow0$, \textit{i.e.}, as the
temperature tends to infinite, the network becomes disconnected and
there is no communication among any pair of nodes. This resembles
a gas in which every node is an independent particle. On the other
hand, when $\beta\rightarrow\infty$, \textit{i.e.}, the temperature
tends to zero, the weights of every edge becomes extremely large,
which definitively increases the communication capacity among the
pairs of connected nodes. The temperature thus plays a role of an
empirical parameter which is useful in simulating effects of external
stresses to which the network is submitted, such as different levels
of social agitation, economical situations, environmental stress,
variable physiological conditions, \textit{etc}. Under this analogy,
we generalized the communicability function~\eqref{eq2.20} into
the form~\cite{Temperature} 
\begin{align}
G_{pq}\left(\beta\right)=\left(e^{\beta A}\right)_{pq}.
\end{align}
It is straightforward to realize that the communicability angle between
a given pair of nodes is generalized to 
\begin{align}
\cos\theta_{pq}\left(\beta\right) & =\frac{G_{pq}\left(\beta\right)}{\sqrt{G_{pp}\left(\beta\right)G_{qq}\left(\beta\right)}}.
\end{align}

Let us conduct a simple experiment to explore the possibilities which
this empirical parameter brings to the analysis of real-world scenarios.
We use two urban street networks representing the city landscapes
of Rio Grande in Brazil and of Yuliang in China. Both cities have
large values of the average communicability angle, \textit{i.e.},
small spatial efficiency, with $\left\langle \theta\right\rangle \approx79.7^\circ$
and $\left\langle \theta\right\rangle \approx85.8^\circ$, respectively.
We then lower the temperature and see if it increases the spatial
efficiency of both cities, \textit{i.e.}, if it decreases the values
of $\left\langle \theta\right\rangle $. In other words, we systematically
increase $\beta$ and compute the average communicability angle $\langle\theta(\beta)\rangle$.
The increase in $\beta$ here can be associated to the average increment
in the number of lanes per street in the city.

Figure~\ref{fig15}(a) shows the results. 
\begin{figure}
\centering 
\includegraphics[width=0.35\textwidth]{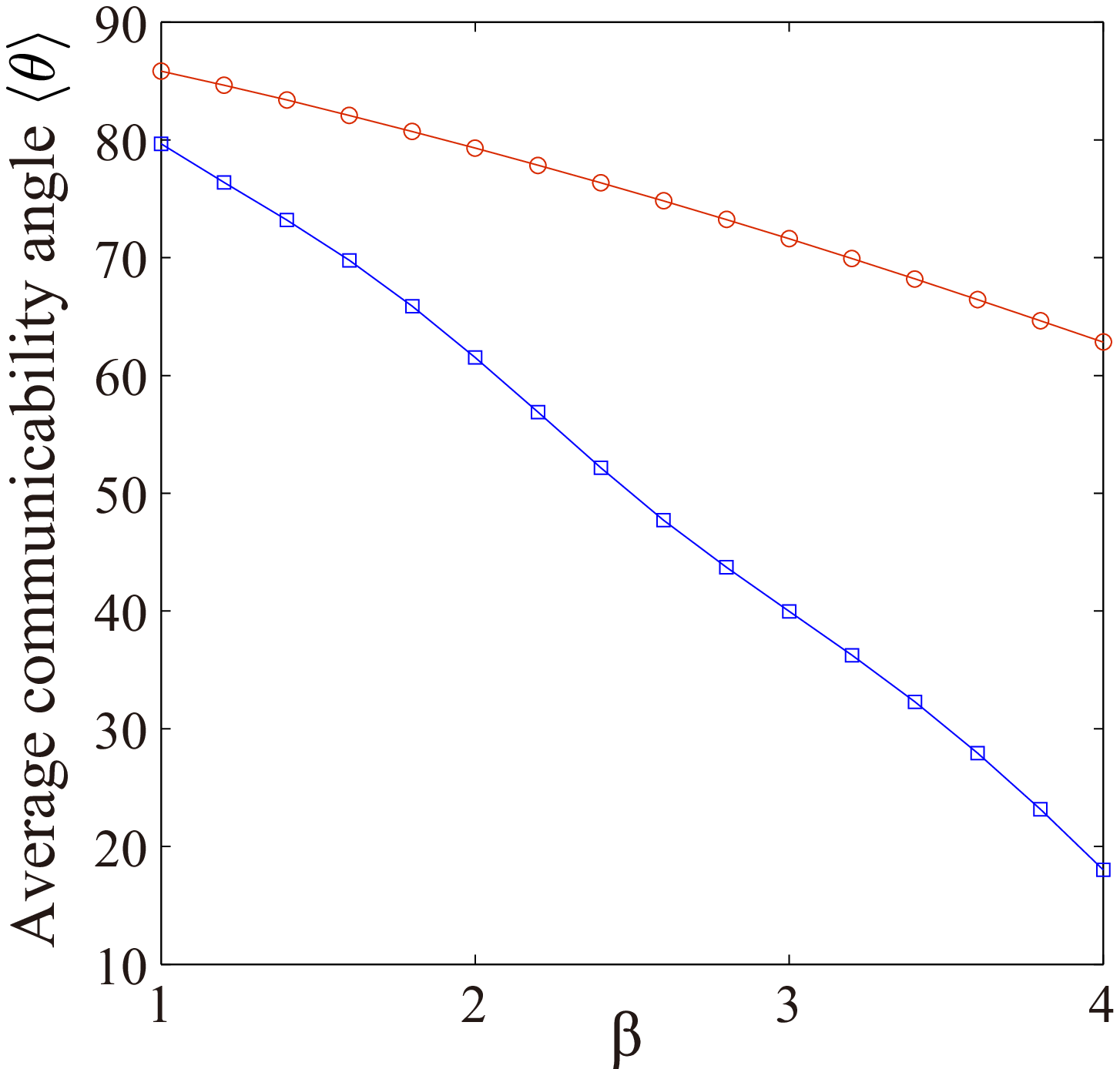}
\hspace*{0.05\textwidth} 
\includegraphics[width=0.35\textwidth]{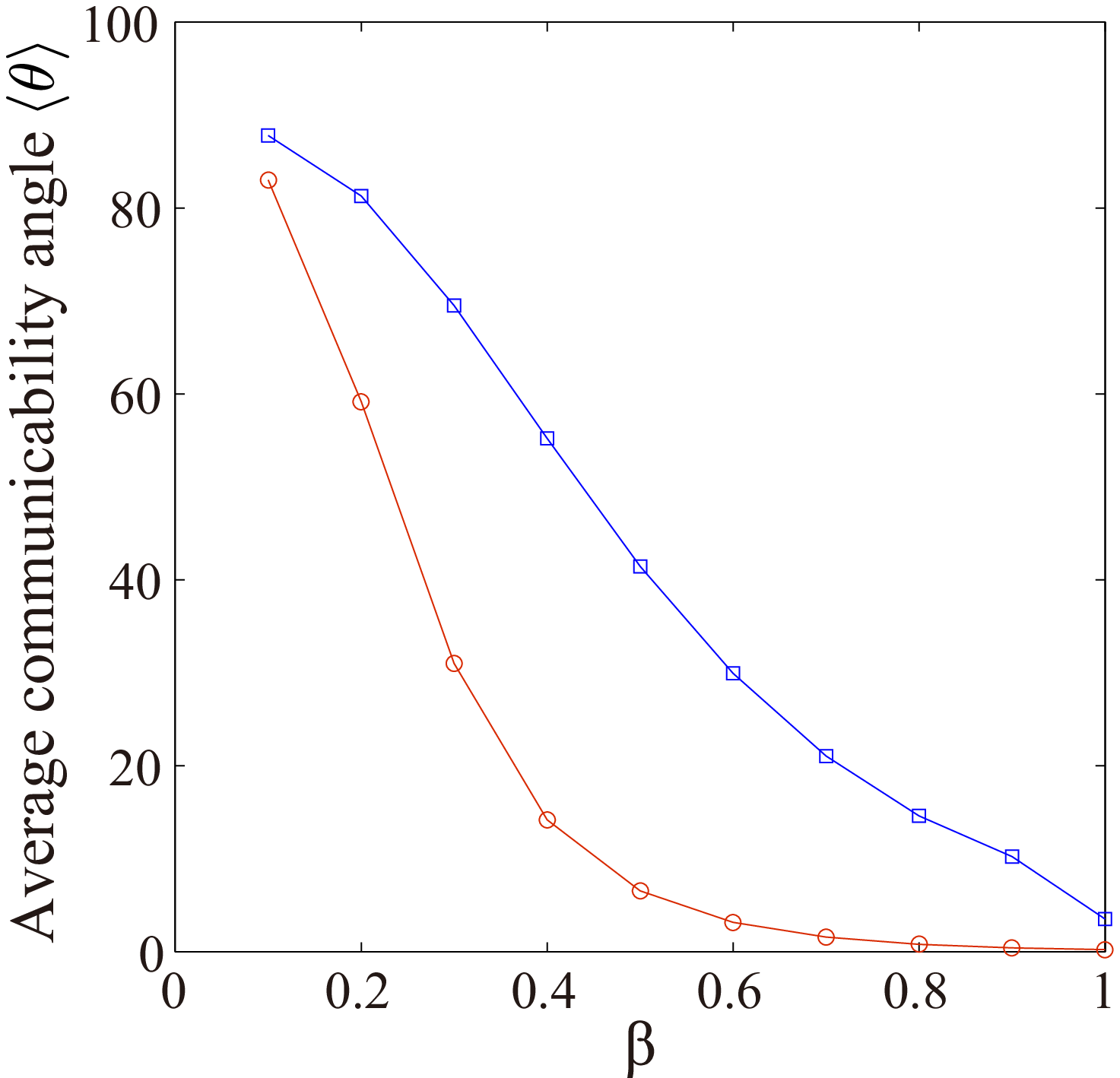}
\caption{(a) Effects of the inverse temperature $\beta$ on the average communicability
angle in two urban street networks, Rio Grande, Brazil (blue
squares) and Yuliang, China (red circles). (b) The
same for two visual-cortex networks: cat (red circles)
and macaque (blue squares).}
\label{fig15} 
\end{figure}
The city of Rio Grande dramatically improves its spatial efficiency
by increasing the average number of lanes of its streets. Although
the improvement for Yuliang is not so dramatic, there is still a decrease
in the average communicability angle of $20^{\circ}$. The causes
for the difference in the variation of $\left\langle \theta\right\rangle $
with the temperature for different networks is not a trivial one,
as there are likely to be many structural factors
involved. We do not investigate these causes here.

We next carry out the opposite experiment using two brain networks
representing the cat and macaque visual cortices. The average communicability
angle shows that both networks have a great spatial efficiency: $\left\langle \theta\right\rangle \approx0.22^\circ$
and $\left\langle \theta\right\rangle \approx3.52^\circ$, respectively.
We here raise the temperature, \textit{i.e.}, decrease $\beta$, and
see if it deteriorates the connections in the visual cortices in terms
of the average communicability angle $\langle\theta(\beta)\rangle$.
The decrease of $\beta$ can be regarded as any malfunctioning or
diseases.

Figure~\ref{fig15}(b) shows the results. Both networks dramatically
decrease their spatial efficiency as $\beta\rightarrow0$; obviously,
$\theta_{pq}\left(\beta=0\right)=90^{\circ}$. We
notice, however, that the cat visual cortex is more resistant to the
stress than the macaque one. For $\beta=0.6$, for example, the former
has $\left\langle \theta\right\rangle \approx3.15^\circ$ while the latter
has jumped up to $\left\langle \theta\right\rangle \approx29.9^\circ$.

The influence of the inverse temperature can be summarized as follows.
In the limit $\beta\rightarrow\infty$, we have $G_{pq}\rightarrow\psi_{1,p}\psi_{1,q}\exp\left(\beta\lambda_{1}\right)$.
This is equivalent to increase the good expansion properties of the
network. We recall from previous sections that for expanders the spectral
gap $\lambda_{1}-\lambda_{2}$ is very large, and consequently we have the above convergence.
This is exactly the effect that we see; 
when we increase $\beta$, the networks become more spatially efficient, \textit{i.e}., 
$\left\langle \theta\right\rangle \rightarrow0^{\circ}$.
In the limit $\beta\rightarrow0$, on the other hand, we have $G_{pq}\rightarrow1$,
which implies that $\left\langle \theta\right\rangle \rightarrow90^{\circ}$.
This is equivalent to reducing dramatically the capacity of each edge
of transmitting information in the network, which clearly decreases
its communication and spatial efficiencies.

In closing, the use of the empirical parameter $\beta$ allows us
to simulate the effects of external factors which can modify the spatial
efficiency of a network. This brings a modeling scenario to assaying
of strategies of improving the spatial efficiency of networks or to
analyses of their resilience to external stresses.

\section{Conclusions}

\label{sec:conclusions}

In a network, the more abstract spatial efficiency
refers to the average quality of communication among the nodes. Such
communication goodness is quantified as the ratio of the amount of
information successfully delivered to its destination to the one which
is frustrated in its delivery and returned to its
originators. This new paradigm is then mathematically formulated in
terms of the communicability angle between a pair of nodes. We have
provided analytical and empirical pieces of evidence which reaffirm
the idea that the communicability angle accounts for the spatial efficiency
of networks.

The richness of this approach goes beyond the results presented here;
there are a few immediate directions of research in this area which
can open new opportunities for the analysis of networks. The use of
the communicability angle for a pair of connected nodes can be seen
as an edge centrality measure which may reveal important characteristics
of individual edges in networks. The communicability angle averaged
over the edges incident to a given node can also represent a node
centrality index which indicates the contribution of the node to the
global spatial efficiency of a network. The study of the effects of
the inverse temperature on the spatial efficiency and the determination
of the most important structural factors that influence it is of tremendous
practical importance. These studies will allow us not only to predict
the effects of external stresses over the spatial efficiency of a
network but also to assay theoretical scenarios of improving this
efficiency in certain classes of networks. Last but not least, the
new concept of communicability angle can bring new possibilities to
the mathematical analysis of specific types of graphs and properties,
such as planarity and graph thickness among others.

\section*{Acknowledgement}

EE thanks the Royal Society for a Wolfson Research Merit Award. He
also thanks Dr.\ Sean Hanna (UCL) for the datasets of urban street
networks used in this work.


\end{document}